\documentclass[a4paper,UKenglish,cleveref, autoref, thm-restate]{lipics-v2021}

\usepackage{bm}
\usepackage{hyperref}
\usepackage{mathtools}
\usepackage{tikz}
\usetikzlibrary{arrows.meta, positioning}
\usetikzlibrary{decorations.pathreplacing}
\usepackage{ragged2e}

\usepackage{pgfplots}
\pgfplotsset{compat=1.18}
\usetikzlibrary{intersections}
\usetikzlibrary{calc}

\newcommand{\rank}{\mathbf{rank}}

\newcommand{\indf}{\mathbf{1}}
\newcommand{\lidx}{\mathcal{I}}
\newcommand{\suppq}{\mathcal{Q}}
\newcommand{\cL}{\mathcal{L}}
\newcommand{\cH}{\mathcal{H}}
\newcommand{\cF}{\mathcal{F}}
\newcommand{\cP}{\mathcal{P}}
\newcommand{\tT}{\widetilde{T}}

\DeclarePairedDelimiter\abs{\lvert}{\rvert}
\DeclarePairedDelimiter\norm{\lVert}{\rVert}

\newcommand{\prob}[1]{\mathbb{P}\left(#1\right)}

\newcommand{\expecs}[2]{\mathbb{E}_{#1}\left[#2\right]}
\newcommand{\expecbs}[2]{\mathbb{E}_{#1}\big[#2\big]}
\def\Inf{\operatornamewithlimits{inf\vphantom{p}}}
\DeclareMathOperator{\spn}{span}
\newcommand{\defeq}{\overset{\mathrm{def}}{=\joinrel=}}

\newcommand{\eps}{\varepsilon}
\newcommand{\pr}{\mathbb{P}}

\renewcommand{\phi}{\varphi}

\let\oldsqrt\sqrt

\def\DHLhksqrt#1#2{%
\setbox0=\hbox{$#1\oldsqrt{#2\,}$}\dimen0=\ht0
\advance\dimen0-0.2\ht0
\setbox2=\hbox{\vrule height\ht0 depth -\dimen0}%
{\box0\lower0.4pt\box2}}

\newcommand{\review}[1]{#1}

\makeatletter
\newcommand\newtag[2]{#1\def\@currentlabel{#1}\label{#2}}
\makeatother

\pdfoutput=1

\title{Lower Bounds for the Algorithmic Complexity of Learned Indexes}

\author{Luis Alberto {Croquevielle}}{Imperial College London, London, United Kingdom}{a.croquevielle22@imperial.ac.uk}{https://orcid.org/0009-0002-0101-4431}{}

\author{Roman Sokolovskii}{Imperial College London, London, United Kingdom}{r.sokolovskii@imperial.ac.uk}{}{}

\author{Thomas Heinis}{Imperial College London, London, United Kingdom}{t.heinis@imperial.ac.uk}{}{}

\authorrunning{L.A. Croquevielle, R. Sokolovskii, T. Heinis}

\Copyright{}



\category{}

\relatedversion{} 

\nolinenumbers

\EventEditors{John Q. Open and Joan R. Access}
\EventNoEds{2}
\EventLongTitle{42nd Conference on Very Important Topics (CVIT 2016)}
\EventShortTitle{CVIT 2016}
\EventAcronym{CVIT}
\EventYear{2016}
\EventDate{December 24--27, 2016}
\EventLocation{Little Whinging, United Kingdom}
\EventLogo{}
\SeriesVolume{42}
\ArticleNo{23}

\begin{document}

\maketitle

\begin{abstract}
    Learned index structures aim to accelerate queries by training machine learning models to approximate the rank function associated with a database attribute. While effective in practice, their theoretical limitations are not fully understood. We present a general framework for proving lower bounds on query time for learned indexes, expressed in terms of their space overhead and parameterized by the model class used for approximation. Our formulation captures a broad family of learned indexes, including most existing designs, as piecewise model-based predictors.

    We solve the problem of lower bounding query time in two steps: first, we use probabilistic tools to control the effect of sampling when the database attribute is drawn from a probability distribution. Then, we analyze the approximation-theoretic problem of how to optimally represent a cumulative distribution function with approximators from a given model class.
    Within this framework, we derive lower bounds under a range of modeling and distributional assumptions, paying particular attention to the case of piecewise linear and piecewise constant model classes, which are common in practical implementations.

    Our analysis shows how tools from approximation theory, such as quantization and Kolmogorov widths, can be leveraged to formalize the space-time tradeoffs inherent to learned index structures. The resulting bounds illuminate core limitations of these methods.
\end{abstract}

\section{Introduction}
\label{sec:1}


Efficient query processing is a fundamental challenge in database systems. For relational databases, key areas of study include join optimization techniques \cite{yannakakis1981algorithms, koutris2025quest}, cardinality estimation methods \cite{atserias2013size, suciu2023applications, leis2015good}, and spatial indexing for multidimensional data \cite{finkel1974quad, guttman1984rtrees}. In general, improvements in query latency involve trade-offs with other resources, such as space overhead or \review{the number of} I/O operations.

For one-dimensional indexing,
tree-based structures such as B+trees \cite{graefe2011modern}
are the predominant choice in practice. These methods aim to answer point and range queries over a single, ordered attribute. Typically, the attribute values (or \textit{keys}) are stored in a sorted structure $A$, and queries are expressed as search operations over $A$. In particular, a point query for a value $q$ requires locating the position of $q$ in $A$, while a range query over an interval $[q, q']$ requires identifying the first key $\geq q$ and scanning forward until a key $> q'$ is reached.

In recent years, machine learning has emerged as a promising tool for query processing. A prominent example for one-dimensional indexing is the Recursive Model Index \cite{kraska2018case}, which maintains a hierarchical structure similar to B+trees but embeds predictive models in internal nodes to estimate key positions within subsets of sorted data.  This work has inspired a broad line of research into ``learned indexes'' \cite{galakatos2019fiting, kipf2020radixspline, hadian2021shift, ding2020alex}, which extend this idea to support dynamic workloads \cite{ding2020alex, wu2021updatable}, streaming data \cite{yang2023flirt, liang2024swix}, and multi-dimensional queries \cite{nathan2020learning, ding2020tsunami, li2020lisa}.

While empirical studies have highlighted their potential advantages over traditional methods \cite{marcus2020benchmarking, sun2023learned}, learned indexes generally lack formal guarantees. Their performance has also been observed to vary across different datasets and query workloads, in contrast to the
consistent behavior of well-established methods such as B+trees. Proving algorithmic guarantees for learned indexes
and characterizing the data-dependent behavior of these methods are all important theoretical questions that have been explored only partially.

In this work, we introduce a formal framework that captures a broad class of learned index structures.
Within this framework, we establish lower bounds on
query time, explicitly characterizing the trade-off between time and space complexity. These bounds hold under general assumptions,
covering a wide range of data-generating processes. \review{As an interesting application, these lower bounds allow us to identify a parameter regime where learned indexes offer no asymptotic advantage over traditional methods}.

The paper is organized as follows: in Section~\ref{sec:2}, we present a more detailed account of related work and our own contributions. Section~\ref{sec:3} introduces the mathematical setting and notation. Section~\ref{sec:4-new} develops our general theoretical approach, while Sections~\ref{sec:6} and~\ref{sec:7} instantiate this framework by applying it to specific function approximation scenarios. Finally, Section~\ref{sec:8} contains concluding remarks and a discussion of future work.



\section{Related Work and Contributions}
\label{sec:2}
The next section introduces our full notation, but we provide some preliminary definitions here. We model an attribute $A$ in a database relation as formed by a sequence $X_1, \ldots, X_n$ of numeric values.
The central object of study in our analysis is the $\rank$ function, defined as $\rank(q) = \sum_{i=1}^n \indf_{(-\infty, X_i]}(q)$, which counts how many keys are less than or equal to a given query value $q$. In the context of learned indexes, query time is typically understood as the time required to evaluate $\rank(q)$.

Classical methods such as B+trees are agnostic to the process by which data is generated. In contrast, learned indexes derive performance guarantees by making assumptions about it, e.g., that the $\{X_i\}$ are i.i.d. with respect to some distribution. This is to be expected in light of related areas of research, in particular statistical learning theory, where generalization guarantees usually rely on i.i.d. assumptions for both the training and testing data.

As a baseline, B+trees support search, insertion, and deletion in $O(\log n)$ time using $O(n)$ space, with worst-case guarantees that hold independently of the data-generating process. Initial theoretical analysis of learned indexes, beginning with the PGM-Index \cite{ferragina2020pgm}, showed performance comparable to B+trees in the static setting, but with $O(\log^2 n)$ search complexity for dynamic workloads. Under (restrictive) assumptions on the data-generating process, a constant improvement in space complexity relative to B+trees was shown \cite{ferragina2020learned}.

More recent work has established that specific learned indexes can asymptotically outperform B+trees in expectation, for the case of static databases. Zeighami and Shahabi \cite{zeighami2023distribution} proved that $O(\log (\log n))$ expected search time is achievable with linear space. Croquevielle et al.~\cite{croquevielle_et_al:LIPIcs.ICDT.2025.19} further improved this to $O(1)$ expected time.
For dynamic databases, Zeighami and Shahabi \cite{zeighami2024theoretical} considered a model where the data distribution can vary over time, with a parameter $\delta$ capturing the amount of distribution shift. They introduce a learned index that supports insertions and queries in $O(\log(\log n))$ expected time when $\delta = 0$ (i.e., no distribution shift), an asymptotic improvement over B+trees. For $\delta > 0$, an extra $O(\log(\delta n))$ cost is incurred, so that asymptotic complexity remains the same as for B+trees.

A complementary line of work was recently introduced by Zeighami and Shahabi \cite{zeighami2024towards}, who study lower bounds for learned database operations from an information-theoretic perspective. In their framework, a learned index is modeled as a parameterized function $f_\theta$ where the parameters $\theta$ are encoded using a total budget of $\sigma$ bits.
The authors derive lower bounds on the model size $\sigma$, showing that if it is too small, there exists at least one input sequence $X_1, \ldots, X_n$ for which a desired level of approximation accuracy cannot be satisfied.

Our work takes a fundamentally different approach. Instead of analyzing a general function $f_\theta$, we focus on a structured class $\lidx$ of indexing algorithms that captures common design elements found in most learned index constructions. While this narrows the scope of our work compared to the general setting of \cite{zeighami2024towards}, it enables sharper lower bounds that are more directly relevant to learned indexes as commonly defined in the literature. Within this class, we establish lower bounds on query time as a function of the space overhead used by the indexing structure.


More specifically, we model learned indexes as directed acyclic graphs (DAGs), where internal nodes direct the search, and sink nodes encode predictive models and store the keys. We develop a general framework for proving query time lower bounds as a function of the model class complexity, the number of sink nodes $K$, the size of the dataset $n$, and the choice of local search algorithm (e.g., linear or binary search).


\review{
    For a broad class of learned indexes, capturing most known designs, we prove that query time is lower bounded by $\Omega(\log(n/K))$ in the worst case. A key implication is that, in order to achieve asymptotic improvement over traditional structures like B+trees (which guarantee $O(\log n)$ query time with linear space), a learned index must use nearly linear space. Specifically, if the number of sink nodes satisfies $K = O(n^\alpha)$ for any $0 < \alpha < 1$, the worst-case query time remains $\Omega(\log n)$---matching that of B+trees. Thus, learned indexes can only hope for worst-case asymptotic gains when space usage is close to linear. When the model class is restricted to piecewise constant functions, this same conclusion further applies to \textit{average-case query time}, where the expectation is taken over the query distribution.
}

\section{Preliminaries}
\label{sec:3}

\subsection{Data as Stochastic Process}

We consider a probabilistic model for a database attribute. Formally, we model the attribute as a stochastic process $X = \{X_i\}_{i \in \mathbb{N}}$, where each $X_i$ is a real-valued random variable. At time $n \in \mathbb{N}$, the attribute consists of the first $n$ values in the sequence, sorted in increasing order. We denote the sorted values by $X_{(1)} \leq \ldots \leq X_{(n)}$, so that the attribute is represented as an array $A_n$ with $A_n[i] = X_{(i)}$ for all $i \in \{1, \ldots, n\}$.
This framework accommodates both worst-case and average-case analyses: we may consider a single sequence $X_1, \ldots, X_n$ to study worst-case behavior, or incorporate the probability measure $\mathbb{P}$ governing the stochastic process to analyze expected performance.

No complexity lower bounds will apply in full generality. For instance, if $X_i = i$, then the rank function $\rank$ can be computed in $O(1)$ time and space, and existing structures such as the PGM-Index easily attain that performance. In this work, we derive lower bounds under the assumption that the $X_i$ are independent random variables drawn according to a common cumulative distribution function (CDF) $F$. We further assume that the query value $q$ is generated independently from the $X_i$ according to a probability measure $\mu$ which
has support $\suppq\subseteq \mathbb{R}$.\footnote{Formally, one must specify measurable spaces for the data and query distributions. For simplicity, we assume that all relevant spaces are equipped with the Borel $\sigma$-algebra and that all functions we integrate are measurable; these assumptions are standard and introduce no meaningful restrictions.
}
We emphasize that independence assumptions are only required for average-case analyses; our worst-case bounds hold unconditionally, though independence is used as a technical tool in their derivation.

Throughout the paper, we use $\expecs{X}{\cdot}$ to denote expectation with respect to the measure $\mathbb{P}$ governing the key sequence $\{X_i\}_{i=1}^n$. The notation $\expecs{q}{\cdot}$ denotes expectation with respect to the query distribution $\mu$, treating the key values $\{X_i\}_{i=1}^n$ as fixed (i.e., conditioning on a realization of the data). When expectations are taken jointly, we write $\expecs{X,q}{\cdot}$, corresponding to integration with respect to the product measure $\mathbb{P} \otimes \mu$.

\subsection{Computational Model and Learning Problem}
\label{sec:3-learning-problem}

As introduced in Section~\ref{sec:2}, we consider
the $\rank$ function over the attribute $A_n$, given by
\begin{equation*}
    \rank_n(q) = \sum_{i=1}^n \indf_{(-\infty, X_i]}(q)\,,
\end{equation*}
where $\indf_U$ denotes the indicator function of a set $U$. Hence, $\rank_n(q)$ maps a query value $q \in \mathbb{R}$ to the number of keys less than or equal to it. Since it is usually clear from context, we omit the sub-index $n$ and write $\rank$ instead. In learned indexes, the goal is to accurately approximate $\rank(q)$ while minimizing the time and space overhead needed to compute it.

We adopt the standard RAM model of computation under the uniform cost criterion \cite{aho1974design}, where each basic instruction requires one time unit and each register, storing an arbitrary integer, uses one space unit. \textit{Time complexity} is defined as the number of basic instructions required to evaluate $\rank(q)$ using an index. \textit{Space complexity} measures the redundant memory used by the index, excluding the space needed to store the attribute $A_n$ itself.

Let $I$ be an indexing procedure that builds a data structure $I(A_n)$ from the attribute $A_n$. For a query value $q$, let $T(I(A_n), q)$ denote the time required to compute $\rank(q)$ using the data structure $I(A_n)$, and let $S(I(A_n))$ denote the space overhead used by $I(A_n)$. Since it is always clear from context, we usually write $T(q)$ in place of $T(I(A_n), q)$.

\review{
    \subsection{Learned Index Structures}
    \label{sec:3-dag-model}
    
    We model a learned index as a directed acyclic graph (DAG), as shown in Figure \ref{fig:dag-structure}. Internal nodes implement routing functions, and the sink nodes contain predictive models and store the keys. Each query follows a path from the root to a sink node, where a prediction of its rank is refined via local search. This structure captures a wide range of learned indexes, most of which include hierarchical routing---such as RMI~\cite{kraska2018case}, ALEX~\cite{ding2020alex}, and PGM-Index~\cite{ferragina2020pgm}.

    For the purpose of lower bounds, we focus on the \textit{local correction step}. Regardless of how queries are routed, the exact rank of $q$ is ultimately recovered by searching around the predicted value $h(q)$ in the sorted array. Since we aim for lower bounds that hold regardless of internal structure---which can vary substantially between different designs---the cost of this local search becomes the main driver of query complexity in our analysis.

    We therefore abstract the index as a piecewise function defined by its sink-node models, and let $K$ denote the number of sinks. This corresponds to a partitioning of the query domain into $K$ regions $\{I_v\}_{v=1}^K$, each associated with a model $h_v$. Although we do not analyze the routing structure, we include it in the model to emphasize that our lower bounds apply regardless of graph topology or the complexity of the routing functions. Understanding how this structure impacts query time remains an open question and may lead to sharper bounds.

    
    We denote by $\lidx$ the class of learned indexes defined in this way, and by $\lidx_{K, \cL}$ those whose sink-node models belong to a function class $\cL$ (e.g., $\cL$ could be the set of polynomial functions) and use at most $K$ sinks. For the rest of the paper, we use $K$ as a lower bound (and proxy) for space usage, and we focus on how the approximation power of $\cL$ interacts with $K$ to characterize performance.
}

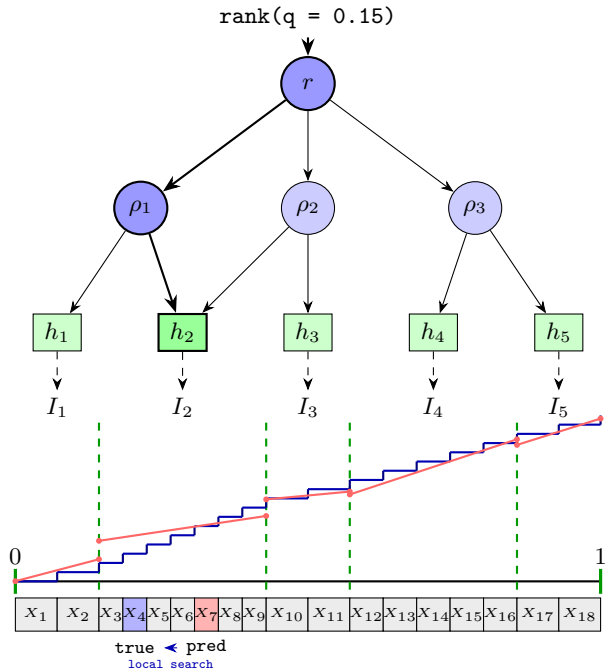
\begin{figure}[ht]
\centering
\begin{tikzpicture}[
    every node/.style={font=\small},
    internal/.style={circle, draw, fill=blue!20, minimum size=1.2em},
    model/.style={rectangle, draw, fill=green!20, minimum width=1.8em, minimum height=1.2em},
    pathnode/.style={thick, fill=blue!40},
    pathmodel/.style={thick, fill=green!40},
    rootnode/.style={circle, draw=black, thick, fill=blue!40, minimum size=2em},
    scale=1.1, >=Stealth
  ]

\node[rootnode] (r) at (0, 6) {$r$};
\node[internal, pathnode] (a) at (-2, 4.5) {$\rho_1$};
\node[internal] (b) at (0, 4.5) {$\rho_2$};
\node[internal] (c) at (2, 4.5) {$\rho_3$};

\node[model] (f1) at (-3, 3) {$h_1$};
\node[model, pathmodel] (f2) at (-1.5, 3) {$h_2$};
\node[model] (f3) at (0, 3) {$h_3$};
\node[model] (f4) at (1.5, 3) {$h_4$};
\node[model] (f5) at (3, 3) {$h_5$};

\node at (0, 6.8) (input) {\texttt{rank(q = 0.15)}};
\draw[->, thick] (input.south) -- (r.north);

\draw[->, thick] (r) -- (a); 
\draw[->] (r) -- (b);
\draw[->] (r) -- (c);

\draw[->] (a) -- (f1);
\draw[->, thick] (a) -- (f2); 
\draw[->] (b) -- (f2);
\draw[->] (b) -- (f3);
\draw[->] (c) -- (f4);
\draw[->] (c) -- (f5);

\draw[thick] (-3.5, 0) -- (3.5, 0);
\foreach \x in {-2.5,-0.5,0.5,2.5} {
  \draw[thick, green!60!black, dashed] (\x,-0.1) -- (\x,2);
}

\node at (-3.5, 0.3) {$0$};
\node at (3.5, 0.3) {$1$};
\draw[very thick, green!60!black] (-3.5,0.15) -- (-3.5,-0.15);
\draw[very thick, green!60!black] (3.5,0.15) -- (3.5,-0.15);

\def\offset{0}
\def\totalkeys{0}
\def\data{2/-3.5/-2.5,7/-2.5/-0.5,2/-0.5/0.5,5/0.5/2.5,2/2.5/3.5}

\foreach \n/\xl/\xr in \data {
  \pgfmathsetmacro{\next}{\totalkeys + \n}
  \xdef\totalkeys{\next}
}
\pgfmathsetmacro{\totalkeys}{\totalkeys / 2}
\xdef\totalkeys{\totalkeys}

\pgfmathsetseed{84}  

\foreach \n/\xleft/\xright in \data {
    \pgfmathsetmacro{\stepwidth}{(\xright - \xleft)/\n}
        
    \foreach \j in {0,...,\numexpr\n-1} {
        \pgfmathsetmacro{\xstart}{\xleft + \stepwidth*\j}
        \pgfmathsetmacro{\xend}{\xleft + \stepwidth*(\j+1)}
        \pgfmathsetmacro{\yval}{(\offset + \j)/\totalkeys}
        \pgfmathsetmacro{\yvalend}{(\offset + \j + 1)/\totalkeys} 
        \pgfmathsetmacro{\keypos}{(\xstart + \xend)/2}
        \pgfmathsetmacro{\keyindex}{\offset + \j + 1}
        \pgfmathtruncatemacro{\keyindex}{\keyindex}
        
        \draw[thick, blue!70!black] ({\xstart}, {\yval}) -- ({\xend}, {\yval});
        \draw[thick, blue!50!black] ({\xend}, {\yval}) -- ({\xend}, {\yvalend});

        \ifnum\keyindex=4
            \draw[fill=blue!30] ({\xstart}, -0.6) rectangle ({\xend}, -0.2);
            \node (X_true) at ({\keypos}, -0.4) {$\scriptscriptstyle X_{\keyindex}$};
        \else\ifnum\keyindex=7
            \draw[fill=red!30] ({\xstart}, -0.6) rectangle ({\xend}, -0.2);
            \node (X_pred) at ({\keypos}, -0.4) {$\scriptscriptstyle X_{\keyindex}$};
        \else
            \draw[fill=gray!15] ({\xstart}, -0.6) rectangle ({\xend}, -0.2);
            \node at ({\keypos}, -0.4) {$\scriptscriptstyle X_{\keyindex}$};
        \fi\fi
    }
    \typeout{rand = rand}
    \pgfmathsetmacro{\vertstart}{rand * (\n / 2) + 0.5}
    \pgfmathsetmacro{\vertend}{rand * (\n / 2) + 0.5}
    \draw[thick, red!60]
    (\xleft,{(\offset - \vertstart)/\totalkeys}) 
        -- (\xright,{(\offset + \n - \vertend) / \totalkeys});
    \fill[red!60] (\xleft,{(\offset - \vertstart)/\totalkeys}) circle[radius=1pt];
    \fill[red!60] (\xright,{(\offset + \n - \vertend) / \totalkeys}) circle[radius=1pt];
    
    \pgfmathsetmacro{\newoffset}{\offset + \n}
    \xdef\offset{\newoffset}
}

\node (true) at ([yshift=-6pt] X_true.south) {\scriptsize \texttt{true}};
\node (pred) at ([yshift=-6pt] X_pred.south) {\scriptsize \texttt{pred}};
\draw[->, blue!70!black] (pred) -- (true)
    node[midway, below=1pt] {\tiny \texttt{local search}};

\draw[->, dashed] (f1.south) -- (-3, 2.3);
\draw[->, dashed] (f2.south) -- (-1.5, 2.3);
\draw[->, dashed] (f3.south) -- (0, 2.3);
\draw[->, dashed] (f4.south) -- (1.5, 2.3);
\draw[->, dashed] (f5.south) -- (3, 2.3);

\node at (-3, 2.1) {$I_1$};
\node at (-1.5, 2.1) {$I_2$};
\node at (0, 2.1) {$I_3$};
\node at (1.5, 2.1) {$I_4$};
\node at (3, 2.1) {$I_5$};

\def\textoffset{6.5cm}
\def\textboxwidth{5.25cm}

\tikzset{
  annotstyle/.style={
    align=left,
    text width=\textboxwidth,
    font=\footnotesize
  }
}

\node[annotstyle] (annot1) at ([xshift=\textoffset] r) {
    \justifying
    \textbf{Root node:} All queries enter here and are routed through the DAG. An example route for $q=0.15$ is highlighted with thicker traversal lines.
};

\node[annotstyle, below=1.2em of annot1] (annot2) {
    \justifying
    \textbf{Internal nodes:} Each internal node $v$ applies a routing function $\rho_v(q)$ to direct the query to one of its children.
};

\node[annotstyle, below=1.2em of annot2] (annot3) {
    \justifying
    \textbf{Sink nodes:} Each sink node $v$ contains a model $h_v(q)$ that estimates $\rank(q)$ within the subinterval $I_v$.
};

\node[annotstyle, below=2em of annot3] (annot4) {
    \justifying
    In this example, the true rank function is shown as a blue step function, and the predictive models $h_v$ are represented as red linear segments.
};

\node[annotstyle, below=1.4em of annot4] (annot5) {
    \justifying
    \textbf{Local search:} The estimate $h_v(q)$ is corrected through a local search over a subarray stored at the sink node.
};

\end{tikzpicture}

\caption{Illustration of a learned index modeled as a directed acyclic graph (DAG).
}
\vspace{-8pt}
\label{fig:dag-structure}
\end{figure}

\section{General Framework for Lower Bounds}
\label{sec:4-new}
We now present a general framework for lower bounding the query time of learned indexes. We consider the query time $T(q)$ as the number of operations needed to compute $\rank(q)$ using the index built on the array $A_n$. The bounds are expressed in terms of the space overhead $S(R(A_n))$, via the number of sink nodes $K$. The analysis is based on three ideas:
\begin{enumerate}[(a)]
    \item\label{overview-a} Relating the query time $T(q)$ to the prediction error $\eps(q) = \abs{h_v(q) - \rank(q)}$. Our analysis focuses on the final step of query execution: the local search after the model prediction. Since this step is always required to confirm or correct the predicted rank, its cost provides a valid lower bound on $T(q)$, \review{which} grows with the magnitude of $\eps(q)$.
    \item\label{overview-b} Reducing the analysis of the prediction error $\eps(q)$ to a function approximation problem for the CDF $F$, by using the approximation $\rank(q) \approx n F(q)$.
    \item\label{overview-c} A probabilistic analysis that formalizes the approximation $\rank(q) \approx n F(q)$ and allows us to derive lower bounds on $T(q)$ under various probabilistic assumptions.
\end{enumerate}
We begin by describing components \eqref{overview-b} and \eqref{overview-c}, and defer discussion of \eqref{overview-a} to the end of this section, where we provide general lower bounds for $T(q)$. For clarity, most proofs are omitted during this exposition, and all omitted technical details are included in the appendix.

\subsection{From Learned Indexes to Approximation Error}
\label{sec:4-learned-indexes-to-cdf}

A learned index as defined in Section~\ref{sec:3-dag-model} can be seen as a structure that approximates $\rank(q)$ through a piecewise-defined prediction function
\begin{equation}
    h(q) = \sum_{v=1}^K h_v(q) \indf_{I_v}(q)\,,
\end{equation}
where each $h_v$ is a predictive model and $I_v \subseteq \suppq$ its associated subdomain. For notational convenience, we identify an index $I$ with its induced predictive model $h$; thus, we take $h\in\lidx_{K,\cL}$ to mean that $h$ is defined piecewise over $K$ disjoint intervals using functions from the class $\cL$.

Given a predictive function $h$, the prediction error is $\eps(q) = |\rank(q) - h(q)|$. To analyze $\eps(q)$ from the perspective of function approximation, we use the \textit{empirical cumulative distribution function} (ECDF). Recall that our standing assumption is that the keys $X_1, \ldots, X_n$ are drawn i.i.d. from a distribution with CDF $F$. The ECDF is correspondingly defined as
\begin{equation*}
    F_n(x) = \frac{1}{n} \sum_{i=1}^n \indf_{(-\infty, X_i]}(x)\,.
\end{equation*}
The rank function can then be expressed as $\rank(q) = n F_n(q)$. Now, for both average and worst-case analysis, it is relevant to consider the expected error over all query values $q$, conditioned on $X_1, \ldots, X_n$. Since $q$ is generated independently from $\{ X_i \}$, we have
\begin{equation}
\label{eq:expected-eps-equals-Fn-approx}
    \expecbs{q}{\eps(q)}
    =
    \int_\suppq \eps(q) d\mu
    =
    \int_\suppq |\rank(q) - h(q)| d\mu
    =
    \int_\suppq |nF_n(q) - h(q)| d\mu
    =
    \norm{nF_n - h}_\mu\,,
\end{equation}
where $\norm{\cdot}_\mu$ denotes the absolute norm in $L^1(\mu)$, the space of $\mu$-integrable functions over $\suppq$. In other words, the best $h$ function on average would be the function that minimizes the absolute distance to $\rank=nF_n$. This motivates the introduction of the \textit{minimum approximation error}:
\begin{equation*}
    R_{K,\cL}(G) = \inf_{h \in \lidx_{K,\cL}} \norm{G - h}_\mu\,,
\end{equation*}
for any $G \in L^1(\mu)$. When $G = \rank = n F_n$, this yields
\begin{equation}
\label{eq:rank-equal-n-times-F}
    R_{K,\cL}(\rank) = R_{K,\cL}(nF_n) = n R_{K,\cL}(F_n)\,,
\end{equation}
by the linear scaling property of the approximation error (see Lemma~\ref{lem:scaling} in Appendix \ref{appendix:A}). This allows us to reduce the analysis of $\expecs{q}{\eps(q)}$ to an approximation problem. The key idea is that $\expecs{q}{\eps(q)}$ can be lower bounded by the optimal $L^1(\mu)$-approximation of $F$, up to a small
deviation term arising from sampling randomness. This deviation is captured by the function $\delta_n(x) = F(x) - F_n(x)$, where $F_n$ is the ECDF.
The following result makes this relationship precise. The proof is straightforward and can be found in Appendix \ref{appendix:A}.

\begin{proposition}
\label{prop:basic-error-lower-bound}
    Let $X_1, \ldots, X_n$ be i.i.d.~samples from a distribution with CDF $F$, and let $q \sim \mu$ independently from the $\{X_i\}$. Then for any predictive model $h \in \lidx_{K,\cL}$ it holds that
    \begin{equation*}
        \expecbs{q}{\eps(q)} 
        =
        \norm{nF_n - h}_\mu 
        \geq
        n \left[R_{K,\cL}(F) - \norm{\delta_n}_\mu\right]\,,
    \end{equation*}
    provided $\cL$ is invariant under scalar multiplication.
\end{proposition}


This inequality shows that the expected prediction error is lower bounded by a deterministic approximation term $n R_{K,\cL}(F)$, minus a data-dependent fluctuation term $n \norm{\delta_n}_\mu$. Controlling this fluctuation is the main probabilistic step in our analysis and will be addressed in the next subsection. In addition to $\expecbs{q}{\eps(q)}$, we also need bounds on $\expecbs{q}{\log_2 \eps(q)}$, which arises in the analysis of exponential search. To this end, we define a proxy for query time:
    $\tT(q) = \log_2 \left(\max\{ 2, \eps(q) \}\right)$,
and establish a lower bound for its expectation. This requires a careful analysis and specific assumptions on $F$, $\mu$, and $\cL$. The following result illustrates the type of bound that can be derived (the proof can be found in Appendix \ref{appendix:prop-proof}).

\begin{proposition}
\label{prop:log-error-lower-bound}
    Let $q \sim \mu$ independently from the keys $\{X_i\}_{i=1}^n$. Assume that the CDF $F$ and $\mu$ have respective densities $f$ and $g$, such that $0<c_F\leq f(q),g(q)\leq C_F<\infty$ for all $q\in\suppq$. Then, there exists a realization of $X_1, \ldots, X_n$ such that:
    \begin{equation*}
        \expecbs{q}{\tT(q)}
        \geq
        C_1 \left[
            \log_2 \left(n R_{K,\cL}(F)\right) - C_2
        \right]\,,
    \end{equation*}
    where $C_1,C_2 > 0$ are independent of $n$ and $K$, and $\cL=P_0$---the class of constant functions.
\end{proposition}

Proposition \ref{prop:log-error-lower-bound} provides a logarithmic analogue to Proposition~\ref{prop:basic-error-lower-bound}, which is useful for analyzing exponential search. The use of $\max\{2, \cdot\}$ in $\tT(q)$ ensures that the proxy for query time is at least $1$.
The constants $C_1$ and $C_2$ depend on $c_F$ and $C_F$, and might be disadvantageous. Still, they are independent of $n$ and $K$, and do not change the asymptotic analysis. Moreover, the condition $0<c_F\leq f,g\leq C_F<\infty$ need not be satisfied on all of $\suppq$; it suffices that it holds over a subinterval $\suppq'\subseteq\suppq$. The main limitation of this variant, in contrast to Proposition \ref{prop:basic-error-lower-bound}, is the constraint $\cL=P_0$, and the fact that this bound only holds for a specific realization $X_1, \ldots, X_n$: it does not necessarily hold for all $\{X_i\}_{i=1}^n \sim \mathbb{P}$.

\subsection{Controlling the Randomness: Concentration of  \texorpdfstring{$F_n$}{Fₙ}}
\label{sec:4-controlling-delta}

The lower bounds in Propositions \ref{prop:basic-error-lower-bound} and \ref{prop:log-error-lower-bound} depend on $R_{K,\cL}(F)$ and the behavior of $\delta_n$. On the one hand, $R_{K,\cL}(F)$ is a deterministic quantity that depends only on $F$, the model class $\cL$, and the measure $\mu$, and can be studied using approximation theory. On the other hand, we can use probabilistic tools to control the term $\norm{\delta_n}_\mu$\,, ensuring it is small enough that the dominant source of error stems from approximating $F$ rather than from sampling variability.

Generally, we wish to say that $\delta_n$ behaves like $1/r(n)$ for some function $r(n)$ that quantifies the rate of convergence of $F_n$ to $F$. This allows us to isolate all randomness in a single term, $1/r(n)$, and express lower bounds on $\expecs{q}{\eps(q)}$ in terms of the deterministic approximation error for $F$. To make this idea precise, we rely on two types of results (all proofs are included in Appendix \ref{appendix:B}): expected and worst-case concentration.


\subparagraph{\textbf{Expected concentration}}
The main result of this type is given by Lemma \ref{lem:expected-norms}, which we will use in conjunction with Proposition \ref{prop:basic-error-lower-bound}.

\begin{lemma}[Expected bounds]
\label{lem:expected-norms}
    Assume $X_1, \ldots, X_n$ are i.i.d. and let $\mu$ be an arbitrary probability measure over $\suppq$, independent from the $\{X_i\}$. Then, it holds that:
    \begin{enumerate}[(a)]
        \item\label{lem:expected-norms-a} $\expecs{X}{\|\delta_n\|_\infty} \leq \sqrt{\frac{\pi}{2n}}$\,.
        \item\label{lem:expected-norms-b} As a consequence, $\expecs{X}{\|\delta_n\|_\mu} \leq \sqrt{\frac{\pi}{2n}}$\,.
    \end{enumerate}
\end{lemma}

These bounds allow us to reason about the concentration of $\delta_n$ in expectation over the random samples $\{X_i\}$, providing an average-case view of the approximation error---again, without making any assumptions on the query distribution $\mu$.

\subparagraph{\textbf{Worst-case concentration}} The main result of this type is given by Lemma \ref{lem:existential-bound}, which we use to analyze the worst-case algorithmic performance of learned indexes.
\begin{lemma}[Existence of small-deviations]
\label{lem:existential-bound}
Assume $X_1, \ldots, X_n$ are i.i.d. and let $\mu$ be an arbitrary probability measure over $\suppq$, independent from the $\{X_i\}$. Then, it holds that:
\begin{enumerate}[(a)]
    \item\label{lem:existential-bound-a} There exists a realization of $\{X_i\}_{i=1}^n$ such that $\norm{\delta_n}_\mu \leq \sqrt{\frac{\pi}{2n}}$\,.
    \item\label{lem:existential-bound-b} If $\mu$ is equal to the data-generating distribution, i.e., $d\mu = dF$, then there exists a realization of $\{X_i\}_{i=1}^n$ such that $\norm{\delta_n}_\mu \leq \frac{1}{\sqrt{6} \, n}$\,.
\end{enumerate}
\end{lemma}

This highlights how stronger assumptions on $\mu$ can yield sharper bounds---specifically, a $1/n$ rate instead of the general $1/\sqrt{n}$. In contrast, bounds that hold for arbitrary $\mu$ must remain more conservative to account for unfavorable scenarios: they may be loose in most cases but cannot be improved without additional assumptions. While a full characterization of how different choices of $\mu$ affect the convergence rate is outside the scope of this paper, we include the $d\mu=dF$ case as it reflects a common setting in the learned indexes literature.

\subsection{From Prediction Error to Query Time}
\label{sec:4-query-time-reduction}

We now complete our framework by translating prediction error into actual query time. As discussed in Section~\ref{sec:3-dag-model}, the final stage of query processing involves a local search to correct the predicted rank $h(q)$ into the true value $\rank(q)$. The cost of this step depends on the prediction error $\eps(q) = |\rank(q) - h(q)|$, and varies depending on the algorithm used. We analyze the three most common search strategies: linear, exponential, and binary search.

\begin{proposition}[Lower bounds for query time via prediction error]
\label{prop:search-error-translation}
Let $q \in \suppq$ be a query with predicted rank $h(q)$ and actual rank $\rank(q)$, and let $\eps(q) = |\rank(q) - h(q)|$ be the prediction error. Then, under the respective search strategies:
\review{
    \begin{enumerate}[(a)]
        \item \textbf{Linear search:}
        \hspace{43pt} $T(q) \geq \eps(q)$\,.
        
        \item \textbf{Exponential search:}
        \hspace{17pt}
        $T(q) \geq \log_2 \left(\max\{ 2, \eps(q) \}\right)$\,.
        
        \item \textbf{Binary search:}
        \hspace{20pt}
        ${\textstyle \sup_q T(q) \geq \log_2 (\sup_q \eps(q))}$\,.
    \end{enumerate}
}
\end{proposition}

\begin{proof}
We analyze each strategy in turn.

\textit{(a)} \textbf{Linear search.} This method scans sequentially from the predicted position $h(q)$ until it reaches the true position $\rank(q)$, as illustrated in Figure \ref{fig:linear-search-diagram}. In practice, the search may advance by 2, 3, or any constant number of positions per step, but this only affects the total number by a constant factor and does not change asymptotic bounds. We therefore assume 1 step per comparison. The total number is thus at least $\eps(q) = |\rank(q) - h(q)|$, so we obtain the lower bound $T(q) \geq \eps(q)$.

\begin{figure}[ht]
\centering
\vspace{-5pt}

\begin{tikzpicture}[every node/.style={scale=0.85},scale=0.8]
\foreach \i in {1,...,12} {
  \draw[fill=gray!15] (\i,0) rectangle ++(1,1);
  \node at (\i+0.5,0.5) {$X_{\i}$};
}
\draw[fill=red!30] (3,0) rectangle ++(1,1);
\node at (3.5,0.5) {$X_{3}$};
\draw[fill=blue!30] (9,0) rectangle ++(1,1);
\node at (9.5,0.5) {$X_{9}$};
\node at (3.5,-0.2) {\small $h(q)$};
\node at (9.5,-0.2) {\small $\rank(q)$};
\foreach \i in {3,...,8} {
  \draw[->, rounded corners=5pt] (\i+0.5,1.05) to[out=60,in=120] (\i+1.5,1.05);
}
\node at (6.5,1.6) {\footnotesize $\eps(q) = 6$ steps};
\end{tikzpicture}

\vspace{-5pt}

\caption{Linear search performs a step-by-step scan from the predicted position $h(q)$ to the true rank $\rank(q)$. The total cost is proportional to the prediction error $\eps(q)$.}

\label{fig:linear-search-diagram}
\end{figure}
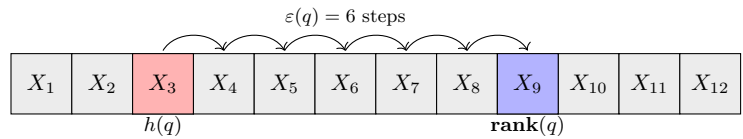

\textit{(b)} \textbf{Exponential search.} Starting at $h(q)$, the algorithm expands geometrically (typically doubling the offset) until it brackets the true position $\rank(q)$ within a window, as visualized in Figure \ref{fig:exp-search-diagram}. It then performs binary search within that window. The number of expansion steps required to reach a position $\geq \rank(q)$ is $\lceil \log_2 \eps(q) \rceil + 1 \geq \log_2 \eps(q)$. The cost of the final binary search is also logarithmic in the window size. Since query time cannot involve less than $1$ operation, we write:
\begin{equation*}
    T(q) \geq \tT(q) = \log_2 \left(\max\{ 2, \eps(q) \}\right),
\end{equation*}
where $\tT$ corresponds to the proxy function introduced earlier.

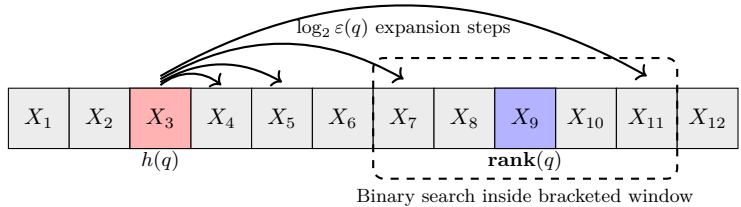
\begin{figure}[ht]
\centering
\vspace{-5pt}

\begin{tikzpicture}[every node/.style={scale=0.85},scale=0.8]
\foreach \i in {1,...,12} {
  \draw[fill=gray!15] (\i,0) rectangle ++(1,1);
  \node at (\i+0.5,0.5) {$X_{\i}$};
}
\draw[fill=red!30] (3,0) rectangle ++(1,1);
\node at (3.5,0.5) {$X_{3}$};
\draw[fill=blue!30] (9,0) rectangle ++(1,1);
\node at (9.5,0.5) {$X_{9}$};
\node at (3.5,-0.2) {\small $h(q)$};
\node at (9.5,-0.2) {\small $\rank(q)$};

\draw[->, thick, rounded corners=4pt] (3.5,1.05) to[out=40,in=140] (4.5,1.05);
\draw[->, thick, rounded corners=4pt] (3.5,1.1) to[out=30,in=150] (5.5,1.1);
\draw[->, thick, rounded corners=4pt] (3.5,1.15) to[out=30,in=150] (7.5,1.15);
\draw[->, thick, rounded corners=4pt] (3.5,1.2) to[out=30,in=150] (11.5,1.2);

\node at (7.5,2) {\footnotesize $\log_2 \eps(q)$ expansion steps};
\draw[thick, dashed, rounded corners] (7, -0.5) rectangle (12, 1.5);
\node at (9.5,-0.8) {\footnotesize Binary search inside bracketed window};
\end{tikzpicture}

\vspace{-5pt}

\caption{Exponential search expands a geometric window around the predicted position $h(q)$ until the true rank $\rank(q)$ is within range. The cost is logarithmic in the prediction error $\eps(q)$.}

\label{fig:exp-search-diagram}
\end{figure}

\textit{(c)} \textbf{Binary search.} In this case, the algorithm relies on a fixed window size that is guaranteed (by construction) to contain the true rank. The size of this window must be chosen to cover the \emph{maximum possible prediction error} over all $q$, i.e., $\sup_q \eps(q)$. Since binary search over a window of size $M$ takes at least $\log_2 M$ steps in the worst case, we obtain:
\begin{equation*}
    {\textstyle \sup_q T(q) \geq \log_2 (\sup_q \eps(q))}\,.
\end{equation*}
This is illustrated in Figure \ref{fig:binary-search-diagram}: although for this specific value of $q$ the prediction error is smaller, the search range must still be wide enough to accommodate the worst-case error.

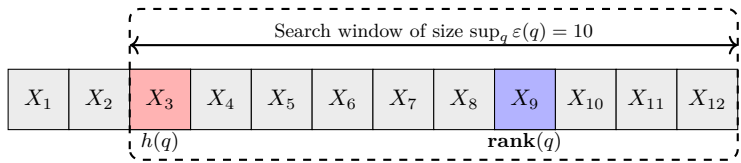
\begin{figure}[ht]
\centering

\begin{tikzpicture}[every node/.style={scale=0.85},scale=0.8]
\foreach \i in {1,...,12} {
  \draw[fill=gray!15] (\i,0) rectangle ++(1,1);
  \node at (\i+0.5,0.5) {$X_{\i}$};
}
\draw[fill=red!30] (3,0) rectangle ++(1,1);
\node at (3.5,0.5) {$X_{3}$};
\draw[fill=blue!30] (9,0) rectangle ++(1,1);
\node at (9.5,0.5) {$X_{9}$};
\node at (3.5,-0.2) {\small $h(q)$};
\node at (9.5,-0.2) {\small $\rank(q)$};

\draw[<->, thick] (3,1.4) -- (13,1.4);
\node at (8,1.6) {\footnotesize Search window of size $\sup_q \eps(q) = 10$};
\draw[thick, dashed, rounded corners] (3, -0.5) rectangle (13, 2);
\end{tikzpicture}

\caption{Binary search requires a search interval guaranteed to contain the true position $\rank(q)$. Therefore, the cost depends on the worst-case prediction error and is logarithmic in the interval size.}

\label{fig:binary-search-diagram}
\vspace{-10pt}
\end{figure}

\end{proof}

\vspace{-10pt}

\subsection{Lower Bounds on Query Time}
\label{sec:4-lower-bounds-query-time}

We are now ready to combine the results from the previous subsections to derive concrete lower bounds on the query time $T(q)$ under the three search strategies described in Section~\ref{sec:4-query-time-reduction}. We consider three natural types of lower bounds, which differ in the strength of their guarantees and the assumptions they require. In increasing order of generality, we consider:
\begin{enumerate}[(A)]
    \item\label{bound-type-A} \textbf{Full expectation bounds (average-case):}
    \begin{equation*}
        \mathbb{E}_{X, q}[T(q)]
        \geq C(n, K, \cL, F, \mu)\,.
    \end{equation*}
    This is the strongest type of result, holding in expectation over both the random dataset $\{X_i\}_{i=1}^n$ and the query value $q \sim \mu$\,. It typically requires strong assumptions (e.g., i.i.d.~sampling) but yields average-case bounds across the entire system.

    \item\label{bound-type-B} \textbf{Conditional expectation bounds (mixed-case):}
    \begin{equation*}
        \sup_{\{X_i\}_{i=1}^n} \mathbb{E}_{q}[T(q)]
        \geq C(n, K, \cL, F, \mu)\,.
    \end{equation*}
    This corresponds to worst-case bound over data and average-case bound over queries.

    \item\label{bound-type-C} \textbf{Pointwise bounds (worst-case):}
    \begin{equation*}
        \sup_{\{X_i\}_{i=1}^n, q \in \suppq} T(q)
        \geq C(n, K, \cL, F, \mu)\,.
    \end{equation*}
    This is a purely worst-case guarantee, asserting that there exists some dataset and query for which the query time is large. It requires the fewest assumptions.
\end{enumerate}

These regimes form a natural hierarchy: a type~\eqref{bound-type-A} bound implies one of type~\eqref{bound-type-B}, which in turn implies one of type~\eqref{bound-type-C}. This implication structure follows directly from the probabilistic method and reflects a trade-off between generality and the strength of assumptions.
Importantly, the bounds hold uniformly over the model class $\lidx_{K,\cL}$\,, i.e., they apply to every learned index structure in $\lidx_{K,\cL}$\,.

Table~\ref{table-lower-bounds-specific} summarizes the resulting lower bounds on query time $T(q)$ under different search strategies, assumptions on the query distribution $\mu$, and the three types of lower bounds \eqref{bound-type-A}--\eqref{bound-type-C} discussed above. These results show how the strongest bounds are attainable for linear search, followed by exponential search, and finally binary search, where we can only derive worst-case guarantees. Formal statements and proofs are deferred to Appendix~\ref{appendix:C}; most follow directly from the general framework developed earlier in this section.

Notice that the type \eqref{bound-type-B} bound for exponential search is written asymptotically, without the constant $C_1$ from Proposition \ref{prop:log-error-lower-bound}, in order to highlight its structural similarity to the rest.

\begin{table}[ht]
\caption{Lower bounds on query time $T(q)$ for various search strategies, bound types, and assumptions on $\mu$. Type \eqref{bound-type-A} bounds further assume that the $\{X_i\}$ are i.i.d. and that $q$ is independent from the $\{X_i\}$. Type \eqref{bound-type-B} bounds assume $\cL=P_0$, that $q$ is independent from the $\{X_i\}$, and the condition $0<c_F\leq f,g\leq C_F<\infty$.}
\label{table-lower-bounds-specific}
\centering
\vspace{-2pt}
\begin{tabular}{c|l|l|l|r}
    \review{\textbf{Bound}} & \textbf{Search} & \textbf{Bound Type} & \textbf{Assumption} & \textbf{Lower Bound} \\ \hline
    \rule{0pt}{15pt}
    \review{\textbf{L1}}
    & Linear      & \eqref{bound-type-A} (avg/avg)  & Any $\mu$
    & $n \big[R_{K,\cL}(F) - \sqrt{\frac{\pi}{2n}}\big]$ \\
    \rule{0pt}{15pt}
    \review{\textbf{L2}}
    & Linear      & \eqref{bound-type-B} (worst/avg)  & $d\mu = dF$ 
    & $n \big[R_{K,\cL}(F) - \frac{1}{\sqrt{6} \, n}\big]$ \\
    \rule{0pt}{15pt}
    \review{\textbf{E1}}
    & Exponential & \eqref{bound-type-B} (worst/avg) & Any $\mu$
    & $\log_2\big(n \,\, R_{K,\cL}(F)\big) \hspace{2pt} - \hspace{2pt} C_2\hspace{4pt}$ \\
    \rule{0pt}{15pt}
    \review{\textbf{E2}}
    & Exponential & \eqref{bound-type-C} (worst/worst) & $d\mu = dF$ 
    & $\log_2\big(n \big[R_{K,\cL}(F) - \frac{1}{\sqrt{6} \, n}\big]\big)$ \\
    \rule{0pt}{15pt}
    \review{\textbf{B1}}
    & Binary      & \eqref{bound-type-C} (worst/worst) & Any $\mu$
    & $\log_2\big(n \big[R_{K,\cL}(F) - \sqrt{\frac{\pi}{2n}}\big]\big)$ \\
    \rule{0pt}{15pt}
    \review{\textbf{B2}}
    & Binary      & \eqref{bound-type-C} (worst/worst) & $d\mu = dF$ 
    & $\log_2\big(n \big[R_{K,\cL}(F) - \frac{1}{\sqrt{6} \, n}\big]\big)$
\end{tabular}
\vspace{-1pt}
\end{table}

These results highlight the trade-offs between model complexity, dataset size, search strategy, and query performance in learned indexes. Note that Table~\ref{table-lower-bounds-specific} considers two representative regimes for the query distribution: (\textit{i}\,) arbitrary $\mu$, and (\textit{ii}\,) the matched case $d\mu = dF$, where queries are drawn from the same distribution as the dataset. The latter is especially relevant in practice, as it is widely used in evaluations of learned indexes~\cite{ferragina2020pgm, marcus2020benchmarking, liang2024swix}.

Altogether, this framework provides a general and flexible method for deriving query time lower bounds.
In the following sections, we analyze the error term $R_{K,\cL}(F)$ for concrete model classes commonly used in learned index structures. Section~\ref{sec:6} focuses on $\cL = P_0$, corresponding to piecewise constant models~\cite{zeighami2023distribution, croquevielle_et_al:LIPIcs.ICDT.2025.19}, while Section~\ref{sec:7} considers $\cL = P_1$, corresponding to piecewise linear approximators, the most widely adopted choice in practice.

\section{Lower Bounds for Piecewise Constant Models}
\label{sec:6}

The general lower bounds developed in Section~\ref{sec:4-new} gain full significance once we derive explicit formulas for the approximation error $R_{K,\cL}(F)$. This requires lower bounding this term for specific choices of the model class $\cL$. 

Virtually all learned index structures rely on polynomial models, as they have few parameters to store, and are fast to train and evaluate. Accordingly, we focus on model classes of the form $\cL = P_m$, where $P_m$ denotes the class of polynomials of degree at most $m$. In this section, we analyze the case $\cL = P_0$, which corresponds to piecewise constant models. Next, in Section~\ref{sec:7}, we address the most widely used case in practice: $\cL = P_1$, corresponding to piecewise linear models.

The key idea in our analysis is that, as we show below, for piecewise constant models the problem of approximating the CDF $F$ reduces to finding the optimal quantization of a random variable distributed over $[0,1]$. This allows us to borrow fundamental results from quantization theory to characterize the approximation error $R_{K,P_0}(F)$.

Our results in this section are derived for a broad class of distributions $F$ and $\mu$---specifically, the case where both are characterized by density functions with respect to the Lebesgue measure. We denote these densities $f$ and $g$, respectively, and consider both the case $f=g$ (Section~\ref{sse:constant_same_distribution}) and the general case (relegated to Appendix~\ref{sse:constant_different_distribution}). In both cases, we show that the approximation error decreases linearly with $K$; the difference is that for $f=g$ this can be shown for any $K$, whereas for $f\neq g$ the result holds asymptotically as $K\to\infty$.


More specifically, throughout Sections \ref{sec:6} and \ref{sec:7}, we focus on the case where $F$ is Lipschitz continuous. This regularity assumption is relatively mild, and it encompasses virtually all common continuous distributions (e.g., Gaussian, exponential, uniform). At the same time, it guarantees the existence of a density $f$, and it ensures that lower bounds reflect the limitations of the model class (as governed by $K$ and $\cL$), rather than artifacts of highly pathological target functions.

\subsection{Same dataset and query distributions \texorpdfstring{$f=g$}{f=g}}
\label{sse:constant_same_distribution}

For the relatively simple case of piecewise constant approximating functions $h\in \lidx_{K, P_0}$ and $f = g$, the minimum approximation error $R_{K,P_0}(F)$ can be calculated exactly for any $K$. We achieve this by first showing that this problem setup reduces to optimal quantization of a uniformly distributed random variable, and then relying on a well-known result in quantization theory to calculate the error of the associated optimal quantizer.

Let $h\in \lidx_{K, P_0}$. Since $F$ is a CDF with range $[0,1]$, without loss of generality we can assume $0\leq h(q)\leq 1$ for all $q$. Therefore, $h$ can be written as
\begin{equation*}
    h(q)
    =
    {\textstyle \sum_{k=1}^K c_k \mathbf{1}_{I_k}(q)}\,,
\end{equation*}
where the $K$ intervals $\{ I_k \}$ partition $\suppq$, and $c_k \in [0,1]$ are the corresponding approximating constants.
This means that when $f=g$, the absolute error $\lVert F - h \rVert_\mu$\ can be written as
\begin{equation}
\label{eq:norm_constant_uniform}
    \lVert F - h \rVert_\mu
    = 
    \int_\suppq
    \lvert F(q) - h(q) \rvert \, f(q) dq
    =
    \sum_{k=1}^K
    \int_{I_k} \lvert F(q) - c_k \rvert \, f(q) dq
    =
    \sum_{k=1}^K
    \int_{J_k} \lvert u - c_k \rvert \, du \,,
\end{equation}
where $J_k = F(I_k) = \{ F(q) : q \in I_k\}$ (which is a continuous interval), and the last equality comes from substituting $u=F(q)$ and $du=f(q)dq$ inside the integrals. Since $F$ is a CDF we know that $\{J_k\}$ is a partition of $[0, 1]$. If we now define
    $Q_K(u)
    =
    \sum_{k=1}^K c_k \mathbf{1}_{J_k}(u)\,,$
then (\ref{eq:norm_constant_uniform}) can be expressed as
\begin{equation}
\label{eq:norm_constant_uniform_expect}
    \lVert F - h \rVert_\mu
    =
    \mathbb{E}\big[ \lvert U - Q_K(U) \rvert \big]\,,
\end{equation}
where $Q_K$ is a piecewise constant function over $[0, 1]$ and the expectation is taken over a $\text{Uniform}([0,1])$ random variable $U$. Note that (\ref{eq:norm_constant_uniform})--(\ref{eq:norm_constant_uniform_expect}) reflect the fact that $F(X)$ is uniformly distributed when $X \sim F$.

In other words, minimizing $\lVert F - h \rVert_\mu$ is equivalent to choosing a set of intervals $\{ J_k \}$ and the associated points $\{ c_k \}$ so as to minimize the expected distance between a realization of $U$ and its corresponding value $c_k$. We recognize this as the problem of optimal quantization of $U$ under the absolute error.

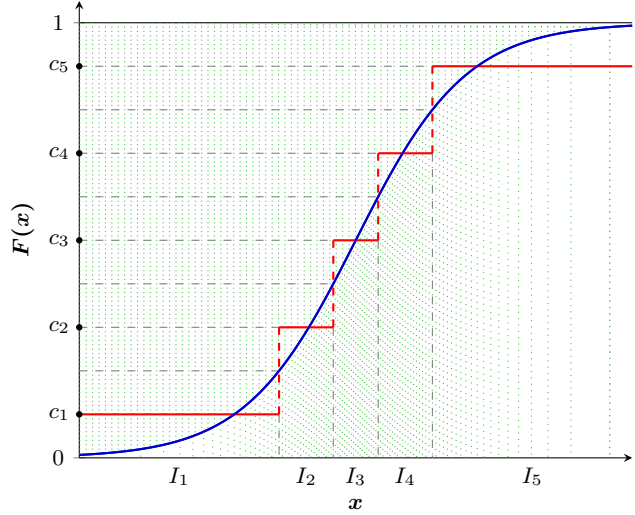
\begin{figure}[t!]
\centering
\newcommand{\maxX}{10}

\begin{tikzpicture}[font=\footnotesize,>=stealth]
\begin{axis}[
    width=252pt,
    xlabel={$\bm{x}$},
    xlabel style={yshift=-10pt},
    ylabel={$\bm{F(x)}$},
    ymin=0, ymax=1.05,
    xmin=0, xmax=\maxX,
    axis lines=left,
    grid=none,
    samples=200,
    ytick={0,1},yticklabels={$0$,$1$},
    xtick=\empty,
    domain=0:\maxX,
    clip=false,
]

\draw (0,1)--(\maxX,1);

\addplot[name path=cdf, thick, blue!80!black] {0.5 + 0.5 * tanh(0.5*(x - 5))};

\pgfmathsetmacro{\M}{100}
\foreach \i in {1,...,\M} {
    \pgfmathsetmacro{\qval}{\i / \M)}

    \edef\temp{\noexpand\path[name path=horizontalline] 
        ({axis cs:0,\qval}) -- ({axis cs:\maxX,\qval});}
    \temp

    \path[name intersections={of=cdf and horizontalline, by={intersection}}];

    \edef\temp{\noexpand\draw[dotted, thin, green!50!gray] ({axis cs:0,\qval}) -- (intersection);}\temp
    \edef\temp{\noexpand\draw[dotted, thin, green!50!gray] (intersection) -- (intersection|-{(0, 0)});}\temp
}

\pgfmathsetmacro{\N}{5}

\foreach \i in {1,...,\N} {
    \pgfmathsetmacro{\qval}{(2*\i - 1) / (2*\N)}
    \pgfmathsetmacro{\qvalmid}{\qval + ( ((2*(\i+1) - 1) / (2*\N)) - ((2*\i - 1) / (2*\N)))/2}

    \pgfmathsetmacro{\qvalprev}{(2*(\i-1) - 1) / (2*\N)}
    \pgfmathsetmacro{\qvalmidprev}{\qvalprev + ( ((2*(\i+1) - 1) / (2*\N)) - ((2*\i - 1) / (2*\N)))/2}


    \edef\temp{\noexpand\path[name path=horizontalline] 
        ({axis cs:0,\qval}) -- ({axis cs:\maxX,\qval});}
    \temp

    \edef\temp{\noexpand\path[name path=horizontallinemid] 
        ({axis cs:0,\qvalmid}) -- ({axis cs:\maxX,\qvalmid});}
    \temp

    \edef\temp{\noexpand\path[name path=horizontallinemidprev] 
        ({axis cs:0,\qvalmidprev}) -- ({axis cs:\maxX,\qvalmidprev});}
    \temp

    \path[name intersections={of=cdf and horizontalline, by={intersection}}];
    \path[name intersections={of=cdf and horizontallinemid, by={intersectionmid}}];
    \path[name intersections={of=cdf and horizontallinemidprev, by={intersectionmidprev}}];

    \edef\temp{\noexpand\draw[dashed, gray] ({axis cs:0,\qval}) -- (intersection);}
    \temp

    \ifnum \i=\N
        \edef\temp{\noexpand\draw[thick,red] (intersectionmidprev|-{(0,\qval)}) -- ({axis cs: \maxX,\qval});}\temp
        \edef\temp{\noexpand\draw[thick,dashed,red] (intersectionmidprev|-{(0,\qvalprev)}) -- (intersectionmidprev|-{(0,\qval)});}\temp
        \coordinate (intersectionmidprevX) at (intersectionmidprev|-{(0,0)});
        \coordinate (intersectionmidX) at (\maxX,0);
    \else
        \edef\temp{\noexpand\draw[dashed, gray] ({axis cs:0,\qvalmid}) -- (intersectionmid);}\temp
        \edef\temp{\noexpand\draw[dashed, gray] (intersectionmid|-{(0,\qval)}) -- (intersectionmid|-{(0, 0)});}\temp
        \ifnum \i=1
            \edef\temp{\noexpand\draw[thick,red] ({axis cs:0,\qval}) -- (intersectionmid|-{(0,\qval)});}\temp
            \coordinate (intersectionmidprevX) at (0,0);
        \else
            \edef\temp{\noexpand\draw[thick,red] (intersectionmidprev|-{(0,\qval)}) -- (intersectionmid|-{(0,\qval)});}\temp
            \edef\temp{\noexpand\draw[thick,dashed,red] (intersectionmidprev|-{(0,\qvalprev)}) -- (intersectionmidprev|-{(0,\qval)});}\temp
            \coordinate (intersectionmidprevX) at (intersectionmidprev|-{(0,0)});
    \fi
        \coordinate (intersectionmidX) at (intersectionmid|-{(0,0)});
    \fi
    \edef\temp{\noexpand\node[below] at ($(intersectionmidprevX)!0.5!(intersectionmidX)$) {$I_{\i}$};}\temp

 

    \edef\temp{\noexpand\filldraw[black]  ({axis cs:0,\qval}) circle (1pt);}
    \temp
    \edef\temp{\noexpand\node[left] at ({axis cs:0,\qval}) {$c_{\i}$};}
    \temp
}

\addplot[name path=cdf, thick, blue!80!black] {0.5 + 0.5 * tanh(0.5*(x - 5))};

\end{axis}
\end{tikzpicture}
\vspace{-2ex}
\caption{Optimal approximation of the CDF $F$ (blue) via a piecewise constant function $h$ (red) with $K=5$ segments when $f=g$---i.e., when the dataset and queries are drawn from the same distribution. In this case, the optimal approximating values $\{ c_k \}$ are evenly spaced (marked on the $y$-axis), and approximating the CDF is equivalent to quantizing $\text{Uniform}([0,1])$ along the $y$-axis. The quantization boundaries (unmarked black horizontal dashed lines) are reflected onto the $x$-axis via $F^{-1}$ to split $\suppq$ into $K=5$ equiprobable regions $\{ I_k \}$ along the quintiles of $F$ (black vertical dashed lines). The green dotted lines illustrate how the distribution of $X\sim F$ becomes uniform when transformed via $F(X)$.}
\vspace{-3ex}
\label{fig:cdf_figure}
\end{figure}

It is a well-established result in quantization theory that uniform quantization is optimal for $\text{Uniform}([0,1])$. Specifically, the quantization points should be equally spaced as $\{ c_k = \frac{2k - 1}{2K} : k = 1,\ldots,K \}$, and the quantization intervals $\{ J_k \}$ should split the domain $[0,1]$ evenly. The associated minimum quantization error is equal to~\cite[Example 5.5]{graf2000foundations}
\begin{equation*}
    T_{K,r}(U) = 
    \inf_{Q_K} \left\{\mathbb{E} \left[ \left|U-Q_K(U)\right|^r \right]  : Q_K\text{~has $K$ quantization points}\right\} =
    \frac{1}{K^r(1+r)2^r} \,;
\end{equation*}
this means that, in our notation,
\begin{equation}
\label{eq:const_abs_error}
    R_{K,P_0}(F) = T_{K,1}(U) = \frac{1}{4 K} \,.
\end{equation}
It is interesting to note that the minimum approximation error~(\ref{eq:const_abs_error}) and the optimal $\{ c_k \}$ and $\{ J_k \}$ are the same for any $F$; the only dependence on $F$ is via the regions $\{ I_k \}$ that should split $\suppq$ at the $K$-quantiles of $F$ to yield equal-sized $\{ J_k \}$.

We illustrate this setting graphically in Figure~\ref{fig:cdf_figure}, where we show an example CDF curve $F(x)$ (blue), optimally approximated by a piecewise constant function $h$ (red) with $K=5$ segments. The approximation points $\{ c_k \}$ are evenly spaced along the $y$-axis, and the mid-points between them, when reflected onto the $x$-axis via $F^{-1}$, mark the quintiles of $F$, which split $\suppq$ into $K=5$ equiprobable regions $\{ I_k \}$. The ``uniformization'' of a distribution by its own CDF that we used in~(\ref{eq:norm_constant_uniform})--(\ref{eq:norm_constant_uniform_expect}) is illustrated by the green dotted lines---their density below the blue CDF curve is proportional to $f$, whereas above the CDF curve, when reflected off it onto the $y$-axis, it becomes uniform.

\section{Lower Bounds for Piecewise Linear Models}
\label{sec:7}
As shown in Section~\ref{sec:6}, learned indexes based on piecewise constant approximations admit a clean lower bound analysis via quantization theory. There, the approximation error $R_{K,\cL}(F)$ was analyzed for individual (Lipschitz continuous) CDFs $F$. In this section, we extend our analysis to more expressive model classes, such as piecewise linear functions. Our central object of study will be the worst-case approximation error over a function class~$\cF$:
\begin{equation*}
    {\textstyle \sup_{F\in\cF} R_{K,\cL}(F)}\,.
\end{equation*}
By lower bounding this quantity, we can show that there exist functions within $\cF$ which are not well approximated by any $h \in \lidx_{K,\cL}$. In line with the rest of the paper, we focus on the class of Lipschitz continuous CDFs, and show that for $\cL = P_1$, the lower bound $\Omega(1/K)$ still holds in this worst-case sense, matching classical upper bounds (see Appendix~\ref{appendix:D-upper-bounds}).

We establish this result through two complementary approaches. In Section~\ref{sec:7-adversarial}, we describe an explicit family of Lipschitz continuous CDFs $\{F_{M}\}$ for which $R_{K,P_1}(F) = \Omega(1/K)$ holds. This shows that the class $\lidx_{K,P_1}$ cannot guarantee $o(1/K)$ approximation error uniformly across all Lipschitz continuous functions, and the construction works for a broad class of measures $\mu$.

In Section~\ref{sec:7-n-widths}, we approach the problem through \textit{Kolmogorov widths}, an important tool from approximation theory. This approach typically assumes $\mu$ as the Lebesgue measure, and only allows for existential (rather than constructive) proofs. Nonetheless, it is a very flexible framework: it allows us to analyze broader model classes beyond $P_1$ and to derive sharper bounds for more regular function spaces. Together, these two methods highlight complementary aspects of the approximation limits faced by learned index structures.

\subsection{An Adversarial CDF Construction}
\label{sec:7-adversarial}

To illustrate the limitations of piecewise linear approximators, we construct a family of CDFs $\{F_M\}_{M\in\mathbb{N}}$ with the following property: for any $M > (1+\alpha)K$ with $\alpha>0$, the approximation error satisfies $R_{K,P_1}(F_M) = \Omega(1/K)$. The construction, detailed in Appendix~\ref{appendix:D-adversarial-cdf}, is based on stitching together rescaled copies of the quadratic function $x^2$ over $M$ disjoint intervals. While this is presented under the assumption that $\mu$ is the uniform distribution, the argument naturally extends to a broader class: specifically, any measure $\mu$ with a density $g$ satisfying $g(q) \geq \lambda > 0$ over an interval. In such cases, the lower bound continues to hold up to constants depending on $\lambda$.

Thus, this construction provides a worst-case lower bound for $R_{K,P_1}(F)$ within the class of Lipschitz continuous CDFs. Importantly, it also highlights a fundamental aspect of learned index structures: the approximation error depends on both the location and the number of difficult-to-approximate regions. Hence, the choice of partition and local model must be adapted to the data. Fixing the model class and budget $K$ in advance cannot guarantee uniformly good approximation.

\subsection{Kolmogorov Widths and Approximation Lower Bounds}
\label{sec:7-n-widths}

To further understand the limitations of learned index models beyond explicit constructions, we turn to a more general framework from approximation theory: Kolmogorov widths. This concept provides a way to quantify how well an entire function class $\cF$ can be approximated by models of bounded complexity. Formally, the Kolmogorov width $d_\kappa(\cF)$ captures the best possible worst-case approximation error over $\cF$:
\begin{equation*}
    d_\kappa(\cF)
    =
    \Inf_{H\in\cH_\kappa} \sup_{F\in\cF} \Inf_{h \in H} \norm{F - h}_\mu \, ,
\end{equation*}
where the outermost infimum is taken over the set $\cH_\kappa$ of all linear subspaces of dimension $\kappa$.

We relate $d_\kappa(\cF)$ to our context in the following way. Throughout, we have considered approximation by functions in $\lidx_{K,\cL}$. This class is not necessarily a linear subspace of $L^1(\mu)$, since the sum of two $K$-piecewise functions is not necessarily $K$-piecewise. Nonetheless, if we fix a partition $\mathcal{P} = \{I_k\}_{k=1}^K$ of $\suppq$, the subclass
\begin{equation*}
    \lidx_{\cP,\cL}
    =
    \Big\{
        h : h(q) = {\textstyle \sum_k} h_k(q) \indf_{I_k}(q),\ h_k \in \cL
    \Big\}
\end{equation*}
does form a linear subspace, provided that $\cL$ is itself a linear subspace of $L^1(\mu)$. In such cases---e.g., when $\cL$ consists of polynomials of bounded degree---any function class $\lidx_{\cP,\cL}$ has dimension at most $\kappa=K \cdot \dim(\cL)$. This characterization leads to the following inequality.

\begin{lemma}
\label{lemma:width-inequality}
For any $\cF, \cL \subseteq L^1(\mu)$ where $\cL$ is a linear subspace of finite dimension, it holds
    \begin{equation*}
        \hspace{65pt}
        \sup_{F\in\cF} R_{K, \cL}(F)
        \geq
        d_\kappa(\cF),
        \hspace{15pt}
        \text{where $\kappa=K\cdot\dim(\cL)$}\,.
    \end{equation*}
\end{lemma}

The proof can be found in Appendix~\ref{appendix:D-proof-fact}. We now focus on the case where $\cF$ is the class of Lipschitz continuous functions, \review{denoted by} $W^{1,\infty}$. This notation comes from Sobolev space theory, which provides a convenient formalism for characterizing function smoothness, and within which most Kolmogorov width results are typically stated. When $\mu$ is the uniform distribution over a compact interval $[a,b]$, it \review{holds that} (see~\cite[Chapter 14, Theorem 3.8]{lorentz1996constructive}):
\begin{equation}
\label{eq:kolmogorov-width-bound}
    d_\kappa(W^{1,\infty}) = \Omega(1/\kappa)\,.
\end{equation}
Notice that for $\cL = P_m$, the class of degree-$m$ polynomials, we have $\kappa = (m+1)K$, since each piece has $m+1$ parameters. Piecewise linear models correspond to $m=1$ and thus $\kappa = 2K$. Hence, for fixed $m$, Lemma \ref{lemma:width-inequality} and equation \eqref{eq:kolmogorov-width-bound} imply $\sup_{F\in W^{1,\infty}} R_{K, P_m}(F) = \Omega(1/K)$.

This has a direct implication for learned index structures: even the best possible model using $K$ adaptive segments and bounded-degree polynomial predictors cannot achieve approximation error $o(1/K)$ for general Lipschitz continuous CDFs.\footnote{This conclusion requires some care. While the Kolmogorov width bound is stated for the worst-case over the class $W^{1,\infty}$, it is not immediate that the worst-case functions realizing the $\Omega(1/\kappa)$ bound include any CDFs. In principle, one might worry that these worst-case functions \review{do not include any} valid CDFs. In Appendix~\ref{appendix:D-lipschitz-cdf}, we provide a rigorous argument showing that any hard-to-approximate function $F \in W^{1,\infty}$ can be transformed into a CDF that remains equally hard to approximate.}

These results match the $\Omega(1/K)$ result obtained in Section~\ref{sec:7-adversarial}. While the explicit construction from Section~\ref{sec:7-adversarial} can accommodate a wide class of distributions $\mu$, including non-uniform densities, the Kolmogorov width approach offers a wider perspective in regards to the approximating functions. In particular, it applies to any piecewise polynomial class $\lidx_{K, P_m}$ and, more generally, to any linear span of $\kappa$ functions, such as truncated Fourier series.

\review{The} Kolmogorov width framework makes it explicit that the $\Omega(1/K)$ bound is inherently a worst-case result. For smoother subclasses $\cF \subseteq W^{1,\infty}$, the theory yields sharper rates: for example, for $W^{2,\infty}$ (i.e., functions with Lipschitz continuous derivatives), one obtains $d_\kappa(W^{2,\infty}) = \Theta(1/\kappa^2)$. This highlights how stronger regularity assumptions can enable faster approximation. In contrast, the lower bound for piecewise constant models (Section~\ref{sec:6}) applies uniformly to all Lipschitz continuous CDFs and does not improve with additional regularity.

\subparagraph{Concluding Remarks.}
A key takeaway from this section is that for the class $\cF = W^{1,\infty}$ of Lipschitz continuous functions, the worst-case approximation error under learned index models in the class $\lidx_{K,P_m}$ satisfies $\Omega(1/K)$ for any fixed $m$. This resembles the $\Omega(1/K)$ rate derived in Section~\ref{sec:6} for piecewise constant approximation, but here it holds only in a worst-case sense over $\cF$, not uniformly for all $F \in \cF$. In particular, the result does not extend to smoother subclasses such as $W^{2,\infty}$, where approximation can be better, but it still constitutes a meaningful lower bound for the general performance of learned indexes.

\section{Discussion and Conclusions}
\label{sec:8}
This paper establishes an analytical framework for lower bounding query time using learned index structures (Sections~\ref{sec:3}--\ref{sec:4-new}). We consider two classes of approximating functions---piecewise constant and piecewise linear approximators---that are most commonly used in practice, and relate query time to the error of approximating the CDF of the data-generating distribution. For piecewise constant functions, we show that the approximation error decreases linearly with the number of segments $K$ for a broad class of CDFs (Section~\ref{sec:6}). For piecewise linear approximators, the situation is more subtle: we show that the approximation error depends on the regularity of the CDF, and that in general using a richer approximating model does not guarantee an improved lower bound on the approximation error (Section~\ref{sec:7}). However, under additional assumptions on the smoothness of the data-generating CDF, the bound on the approximation error can be lowered further.

\review{
    More specifically, Sections~\ref{sec:6} and~\ref{sec:7} establish a lower bound of $\Omega(1/K)$ on approximation error that, in general, cannot be improved for the class of Lipschitz continuous CDFs. For piecewise constant models, this holds uniformly (Section~\ref{sec:6}); for piecewise linear models, the bound applies in a worst-case sense over the choice of CDF (Section~\ref{sec:7}). In either case, the $\Omega(1/K)$ behavior governs the worst-case regime. Combining these results with the analysis summarized in Table~\ref{table-lower-bounds-specific} (e.g., see rows E2 and B2) we obtain an $\Omega(\log(n/K))$ lower bound on worst-case query time. As discussed in Section~\ref{sec:2}, this implies that the number of segments $K$ (and therefore, space usage) must be nearly linear in $n$ to yield an asymptotic advantage over traditional methods, a conclusion that holds for a broad class of learned indexes.
}

It is instructive to compare our bounds with the information-theoretic results in~\cite{zeighami2024towards}. Their worst-case bounds require a finite query domain $\suppq$, a restriction not present in our framework. Moreover, our analysis yields sharper and more applicable insights for typical learned index designs, particularly those based on piecewise constant or linear models. In the average-case setting, their results imply that achieving expected error $\eta$ requires $\log(1/\eta)$ bits of model capacity. In contrast, our bounds---under independence assumptions---show a $1/\eta$ scaling, which is larger when $\eta \to 0$. Additionally, when the query and data distributions are the same---i.e., $d\mu=dF$---we show that the number of parameters can scale as $\Omega(n)$, compared to $\Omega(\sqrt{n})$ in their framework.


Aside from their theoretical importance, our results offer some practical insights for the design of learned indexes: first, they show that the number of segments $K$ should not be chosen without regard to the underlying data, as there exist ``adversarial'' distributions for which the approximation error can be poor for a given $K$ (see Section~\ref{sec:7-adversarial}). Second, the results of Section~\ref{sec:7-n-widths} show that approximation capacity depends not only on the expressiveness of the model class but also on the smoothness of the CDF. This suggests that selecting the model class in a data-dependent manner---adapted to the regularity of the CDF---could lead to more efficient index designs, an idea that warrants further exploration.

Several open directions remain: extending average-case analysis to binary search; relaxing independence assumptions on the data; considering a wider class of index structures; and analyzing learned index performance under updates or dynamic workloads. We hope this work offers some of the foundational tools needed to explore these ideas further.






\bibliography{main}

\appendix

\section{Proof of Scaling Lemma and Proposition \ref{prop:basic-error-lower-bound}}
\label{appendix:A}
\begin{lemma}[Scaling of Approximation Error]
\label{lem:scaling}
    Let $\lambda > 0$, and suppose the model class $\cL$ is invariant under scalar multiplication. Then for any function $G \in L^1(\mu)$ and any $K \in \mathbb{N}$,
    \begin{equation*}
        R_{K,\cL}(\lambda G) = \lambda \cdot R_{K,\cL}(G)\,,
    \end{equation*}
    where $R_{K,\cL}(G) = \inf_{h \in \lidx_{K,\cL}} \norm{G - h}_\mu$.
\end{lemma}

\begin{proof}
    Let $\mathcal{H} = \lidx_{K,\cL}$ denote the class of piecewise models defined over $K$ intervals and using functions from $\cL$. Since $\cL$ is invariant under scalar multiplication, so is $\mathcal{H}$. That is, for any $h \in \mathcal{H}$ and any $\lambda > 0$, the function $\lambda h$ also belongs to $\mathcal{H}$. Now observe:
    \begin{align*}
        R_{K,\cL}(\lambda G)
        &= \inf_{h \in \mathcal{H}} \norm{\lambda G - h}_\mu
        = \inf_{h \in \mathcal{H}} \norm{\lambda (G - \lambda^{-1} h)}_\mu
        = \inf_{h \in \mathcal{H}} \lambda \cdot \norm{(G - \lambda^{-1} h)}_\mu\,.
    \intertext{Using the positive homogeneity of the infimum:}
        R_{K,\cL}(\lambda G)
        &= \lambda \cdot \inf_{h \in \mathcal{H}} \norm{G - \lambda^{-1} h}_\mu\,.
    \intertext{Since $\mathcal{H}$ is invariant under scalar multiplication, it follows that for any $h \in \mathcal{H}$, we have $\lambda^{-1} h \in \mathcal{H}$. Moreover, any $h \in \mathcal{H}$ can be written as $\lambda^{-1} \tilde{h}$ where $\tilde{h}\in\mathcal{H}$. This implies that the map $h \mapsto \lambda^{-1} h$ is a bijection on $\mathcal{H}$, and thus}
        R_{K,\cL}(\lambda G)
        &= \lambda \cdot \inf_{\tilde{h} \in \mathcal{H}} \norm{G - \tilde{h}}_\mu
        = \lambda \cdot R_{K,\cL}(G)\,.
    \end{align*}
\end{proof}

\begin{proof}[Proof of Proposition \ref{prop:basic-error-lower-bound}]
Applying the triangle inequality in $L^1(\mu)$, we have for any $h \in \lidx_{K,\cL}$:
\begin{equation*}
    \norm{F - h}_\mu 
    \leq
    \norm{F - F_n}_\mu + \norm{F_n - h}_\mu
    \quad \implies \quad
    \norm{F_n - h}_\mu 
    \geq
    \norm{F - h}_\mu - \norm{\delta_n}_\mu\,.
\end{equation*}
Taking the infimum over $h \in \lidx_{K,\cL}$ at the last inequality gives:
\begin{equation*}
    R_{K,\cL}(F_n) \geq R_{K,\cL}(F) - \norm{\delta_n}_\mu\,.
\end{equation*}
Then, by the scaling identity $R_{K,\cL}(nF_n) = n R_{K,\cL}(F_n)$ (Lemma~\ref{lem:scaling}) and the definition of expected prediction error $\expecbs{q}{\eps(q)} = \norm{nF_n - h}_\mu \geq R_{K,\cL}(nF_n)$, the result follows.
\end{proof}

\section{Proof of Proposition \ref{prop:log-error-lower-bound}}
\label{appendix:prop-proof}
The focus of this section is to prove Proposition \ref{prop:log-error-lower-bound}. For clarity, some auxiliary results used in this section are proven in Appendix \ref{appendix:prop-aux}. We start by recalling the definition of $\tT(q)$:
\begin{equation*}
    \tT(q)
    = \log_2 \left(\max\{ 2, \eps(q) \}\right)
    = \log_2 \left(\max\{ 2, \abs{nF_n(q) - h(q)} \}\right)\,.
\end{equation*}

Ultimately, we are interested in lower bounding $\expecs{q}{\tT(q)}$. To this end, recall first that (by hypothesis) $F$ and $\mu$ have respective densities $f$ and $g$ which satisfy
\begin{equation}
\label{eq:bounds-densities}
    0<c_F\leq f(q),g(q) \leq C_F < \infty
\end{equation}
for constants $c_F, C_F$, and for all $q\in\suppq$. Proposition \ref{prop:log-error-lower-bound} (and this proof) can be generalized straightforwardly to the case where this property is only satisfied in a continuous interval $\suppq'\subseteq\suppq$ with positive measure. Nonetheless, for simplicity of exposition, we assume \eqref{eq:bounds-densities} holds for all $q\in\suppq$.

Note that without loss of generality, we may assume $d\mu = dF=f\,dq$. This follows from:
\begin{equation*}
    \expecs{q}{\tT(q)}
    =
    \int_\suppq \tT(q) d\mu
    =
    \int_\suppq \tT(q) \, g(q) dq
    \geq 
    \frac{c_F}{C_F} \int_\suppq \tT(q) \, f(q) dq \, ,
\end{equation*}
where we have used $d\mu=g\,dq$ and the inequalities expressed in \eqref{eq:bounds-densities}. Hence, if we can lower bound the last integral $\int_\suppq \tT(q)\,f(q)dq$ as stated in Proposition \ref{prop:log-error-lower-bound}, the bound will also apply to $\expecs{q}{\tT(q)}$ with an extra (constant) factor of $\frac{c_F}{C_F}$. Hence, throughout this section we write $dF$ in place of $d\mu$.

Finally, we assume $nR_{K,P_0}\geq 1$. Otherwise, Proposition \ref{prop:log-error-lower-bound} holds trivially since its right hand side is negative. With these assumptions in place, we can proceed to the proof.

\begin{proof}[Proof of Proposition \ref{prop:log-error-lower-bound}]
We divide the proof into four logical steps.

\subparagraph{Step 1: Piecewise expansion.} Any $h\in \lidx_{K,P_0}$ can be expanded as
\begin{equation*}
    h(q)
    = 
    n\sum_{k=1}^K c_k\indf_{I_k}(q),
\end{equation*}
where the $\{c_k\}$ are constants and the $\{I_k\}$ are intervals which partition $\suppq$. Now, the quantity of interest can be lower bounded as:
\begin{align}
    \expecbs{q}{\tT(q)}
    &=
    \int_\suppq
    \log_2\left(\max\{2, \eps(q)\}\right) dF \nonumber \\
    &=
    \int_\suppq
    \log_2\left(\max\big\{
        2, \abs{nF_n(q) - h(q)}
    \big\}\right) dF \nonumber \\
    &\geq
    \int_\suppq
    \log_2\left(\max\big\{
        2, n\abs{F(q) - n^{-1}h(q)} - n\abs{\delta_n(q)}
    \big\}\right) dF \nonumber \\
    &=
    \sum_{k=1}^K
    \underbrace{
        \int_{I_k}
        \log_2\left(\max\big\{
            2, n\abs{F(q) - c_k} - n\abs{\delta_n(q)}
        \big\}\right) dF
    }_{E_k} \label{eq:prop-log-lower-bound-1} \, ,
\end{align}
where we have used the triangle inequality to lower bound $\eps(q)$:
\begin{align*}
    \eps(q)
    &=
    \abs{nF_n(q) - h(q)} \\
    &=
    \abs{nF_n(q) - nF(q) + nF(q) - h(q)} \\
    &\geq
    n\abs{F(q) - n^{-1}h(q)}
    -
    n\abs{F_n(q) - F(q)} \\
    &=
    n\abs{F(q) - n^{-1}h(q)}
    -
    n\abs{\delta_n(q)} \, .
\end{align*}
Now, without loss of generality, we can assume that $h$ intersects with $F$ in all intervals, that is, for each $k$ there exists $q_k\in I_k$ such that $F(q_k)=c_k$. Otherwise, there would exist a different $h'\in \lidx_{K,P_0}$ such that $\abs{F(q) - n^{-1}h'(q)} \leq \abs{F(q) - n^{-1}h(q)}$ for all $q\in\suppq$, meaning
\begin{align*}
    \expecbs{q}{\tT(q)}
    &\geq
    \int_\suppq
    \log_2\left(\max\big\{
        2, n\abs{F(q) - n^{-1}h(q)} - n\abs{\delta_n(q)}
    \big\}\right) dF \\
    &\geq
    \int_\suppq
    \log_2\left(\max\big\{
        2, n\abs{F(q) - n^{-1}h'(q)} - n\abs{\delta_n(q)}
    \big\}\right) dF \, ,
\end{align*}
so that any lower bound we derive for the case of $h'$ would also apply to $h$. Hence, we assume the existence of $\{q_k\}_{k=1}^K$ with the above property.

From these derivations, it is clear that it is sufficient for our proof to lower bound the expression $\sum_{k=1}^K E_k$ as defined in \eqref{eq:prop-log-lower-bound-1}, where each $c_k$ is assumed to intersect $F$ in at least one point $q_k\in I_k$. In the next part of the proof, we focus on lower bounding each individual term $E_k$, given by
\begin{equation}
\label{eq:ek-lower-bound-1}
    E_k
    =
    \int_{I_k}
    \log_2\left(\max\big\{
        2, n\abs{F(q) - c_k} - n\abs{\delta_n(q)}
    \big\}\right) dF \, .
\end{equation}
To do this, we first define a setting in which the term $n\abs{\delta_n(q)}$ is negligible, and then derive some concentration results for the remaining term $n\abs{F(q)-c_k}$. Finally, in the last part of the proof, we turn to the aggregate error $\sum_{k=1}^K E_k$.

\subparagraph{Step 2: A large-measure set where $\delta_n$ is small.}
Now, define the set
\begin{equation*}
    A
    =
    \Big\{
        q\in\suppq \>\text{ such that }\>
        \abs{\delta_n(q)} \leq \frac{1}{n}
    \Big\} \, .
\end{equation*}
Since $d\mu=dF$, by Lemma \ref{lem:existential-bound}\eqref{lem:existential-bound-b} there exists a realization of $X_1,\ldots,X_n$ such that $\norm{\delta_n}_\mu \leq \frac{1}{\sqrt{6} \, n}$. We work under this realization for the rest of the proof. By Markov's inequality:
\begin{equation*}
    \mu\left(\abs{\delta_n} > \frac{1}{n}\right)
    \leq 
    \frac{\norm{\delta_n}_\mu}{\frac{1}{n}}
    \leq
    \frac{n}{\sqrt{6} \, n}
    \leq
    \frac{1}{2} \, ,
\end{equation*}
meaning that $\mu(A)\geq 1/2$. Now, for each $k=1,\ldots,K$ define the set $J_k=I_k \cap A$, that is, the subset of interval $I_k$ where $\abs{\delta_n} \leq \frac{1}{n}$ holds. Starting from \eqref{eq:ek-lower-bound-1}, we have
\begin{align}
    E_k
    &=
    \int_{I_k}
    \log_2\left(\max\big\{
        2, n\abs{F(q) - c_k} - n\abs{\delta_n(q)}
    \big\}\right) dF \nonumber \\
    &\geq
    \int_{J_k}
    \log_2\left(\max\big\{
        2, n\abs{F(q) - c_k} - n\abs{\delta_n(q)}
    \big\}\right) dF \nonumber \\
    &\geq
    \frac{1}{2}
    \int_{J_k}
    \log_2\left(
        2 + n\abs{F(q) - c_k} - n\abs{\delta_n(q)}
    \right) dF \nonumber \\
    &\geq
    \frac{1}{2}
    \int_{J_k}
    \log_2\left(
        1 + n\abs{F(q) - c_k}
    \right) dF \label{eq:prop-log-lower-bound-2}\, ,
\end{align}
where we have used the fact that $\abs{\delta_n} \leq \frac{1}{n}$ over each set $J_k$.

\subparagraph{Step 3: A concentration result for the approximation error.}

We now focus on lower bounding expression \eqref{eq:prop-log-lower-bound-2}. For this, we define the first and second moments of $\abs{F(q) - c_k}$ over the set $J_k$:
\begin{equation*}
    m_k = \int_{J_k} \abs{F(q) - c_k}\,dF \, ,
    \quad\quad
    s_k = \int_{J_k} \abs{F(q) - c_k}^2 dF \, .
\end{equation*}
Next, denote by $J_k'$ the set $J_k'=\big\{q\in J_k \>\text{such that}\> \abs{F(q) - c_k} \geq \frac{m_k}{2\mu(J_k)}\big\}$, that is, the subset of $J_k$ where the quantity $\abs{F(q) - c_k}$ is not much smaller than its normalized average over $J_k$. Since $\log_2\left(1 + n\abs{F(q) - c_k}\right)$ is never negative, it holds that
\begin{align}
    E_k
    &\geq
    \frac{1}{2}\int_{J_k}
    \log_2\left(
        1 + n\abs{F(q) - c_k}
    \right) dF \nonumber \\
    &\geq
    \frac{1}{2}\int_{J_k'}
    \log_2\left(
        1 + n\abs{F(q) - c_k}
    \right) dF \nonumber \\
    &\geq
    \frac{\mu(J_k')}{2}
    \log_2\left(
        1 + \frac{n m_k}{2\mu(J_k)}
    \right) \nonumber \\
    &\geq
    \frac{m_k^2}{8 \, s_k}
    \cdot
    \log_2\left(
        1 + \frac{n m_k}{2\mu(J_k)}
    \right) \label{eq:prop-log-lower-bound-3}\, ,
\end{align}
where the last step comes from the Paley-Zygmund inequality (see Lemma \ref{lem:paley-zygmund-general}). Therefore, the only remaining challenge is to lower bound the quantity $m_k$ and upper bound $s_k$. To that end, we leverage the bi-Lipschitz nature of $F$: since $f$ is bounded as $c_F\leq f \leq C_F$, this directly implies
\begin{equation}
\label{eq:prop-log-lower-bound-4}
    c_F\abs{x - y}
    \leq \abs{F(x) - F(y)}
    \leq C_F\abs{x - y}
\end{equation}
for all $x,y \in \suppq$. Now, recall our assumption that for each $k$, there exists $q_k \in I_k$ such that $F(q_k) = c_k$. Hence, we may turn to inequality \eqref{eq:prop-log-lower-bound-4} with $y = q_k$ to get, for all $q \in I_k$:
\begin{align*}
    \abs{F(q) - c_k}
    &=
    \abs{F(q) - F(q_k)}
    \geq
    c_F\abs{q - q_k} \, , \\
    \abs{F(q) - c_k}
    &=
    \abs{F(q) - F(q_k)}
    \leq
    C_F\abs{q - q_k} \, .
\end{align*}
Now, under these last conditions, we may invoke Lemma \ref{lem:bounds-first-second-moment} to obtain bounds on $m_k$ and $s_k$:
\begin{equation*}
    m_k \geq \frac{c_F^2}{4C_F^2} \, \mu(J_k)^2 \, ,
    \quad\quad
    s_k \leq \frac{C_F^3}{3c_F^3} \, \mu(I_k)^3 \, .
\end{equation*}
We now substitute these values into the Paley-Zygmund bound from \eqref{eq:prop-log-lower-bound-3}, yielding:
\begin{equation}
\label{eq:ek-lower-bound-2}
    E_k
    \geq
    3
    \left(\frac{c_F}{2C_F}\right)^7
    \frac{\mu(J_k)^4}{\mu(I_k)^3}
    \log_2\left(
        1 + \frac{1}{2}\left(\frac{c_F}{2C_F}\right)^2 n\,\mu(J_k)
    \right) \, .
\end{equation}
This provides an explicit, fully non-asymptotic lower bound for each $E_k$ in terms of $c_F, C_F, n$, and the measures $\mu(J_k)$ and $\mu(I_k)$. For the last part of the proof, we incorporate further information about these measures and aggregate over all $k=1,\ldots,K$.

\subparagraph{Step 4: Summing over $k$ and final aggregation.} We first partition the sets $\{J_k\}$ into two distinct groups, based on the extent to which their measures have decreased relative to the corresponding original $I_k$. Specifically, we define the index set $\mathcal{K}$ as
\begin{equation*}
    \mathcal{K}
    =
    \{
        k\in\{1,\ldots,K\}
        \>\text{ such that }\>
        \mu(J_k) \geq \mu(I_k) / 3
    \} \, .
\end{equation*}
A key observation is that, when summing over all $k$ outside of $\mathcal{K}$, we get
\begin{equation*}
    \sum_{k\notin \mathcal{K}}\mu( J_k )
    <
    \frac{1}{3}
    \sum_{k\notin \mathcal{K}}\mu( I_k )
    \leq
    \frac{1}{3} \mu(\suppq)
    =\frac{1}{3} \, .
\end{equation*}
From this, we derive a key condition:
\begin{equation}
\label{eq:prop-log-lower-bound-5}
    \sum_{k\in \mathcal{K}}\mu( J_k )
    =
    \sum_{k=1}^K\mu( J_k )
    -
    \sum_{k\notin \mathcal{K}}\mu( J_k )
    =
    \mu(A)
    -
    \sum_{k\notin \mathcal{K}}\mu( J_k )
    \geq
    \frac{1}{2} - \frac{1}{3}
    = \frac{1}{6} \, ,
\end{equation}
where we have used the fact that $\mu(A) \geq 1/2$, as derived earlier via Markov's inequality. Notice that this necessarily means $\abs{\mathcal{K}}$ is not empty; hence, we can assume $\abs{\mathcal{K}}\geq 1$. Now, by definition, for any $k\in\mathcal{K}$ it holds that $3\mu(J_k)\geq \mu(I_k)$. Combining with \eqref{eq:ek-lower-bound-2} and denoting $\gamma=\frac{c_F}{2C_F}$:
\begin{align*}
    \sum_{k=1}^K E_k
    &\geq
    \sum_{k\in \mathcal{K}} E_k \\
    &\geq
    \frac{\gamma^7}{9}
    \sum_{k\in \mathcal{K}} 
    \mu(J_k)
    \log_2\left(
        1 + \frac{\gamma^2 n}{2}\,\mu(J_k)
    \right) \\
    &\geq
    \frac{\gamma^7}{9}
    \sum_{k\in \mathcal{K}} 
    \mu(J_k)
    \log_2\left(
        \frac{\gamma^2 n}{2}\,\mu(J_k)
    \right) \\
    &=
    \frac{\gamma^7}{9}
    \left[
        \log_2\left(\frac{\gamma^2 n}{2}\right)
        \sum_{k\in \mathcal{K}} \mu(J_k)
        +
        \sum_{k\in \mathcal{K}}
        \mu(J_k)\log_2\left(\mu(J_k)\right)
    \right]\, .
\intertext{
    Denote by $S$ the sum $S=\sum_{k\in \mathcal{K}} \mu(J_k)$. Then, by Lemma \ref{lem:partial-entropy-bound}, we get:
}
    \sum_{k=1}^K E_k
    &\geq
    \frac{\gamma^7}{9}
    \left[
        S\log_2\left(\frac{\gamma^2 n}{2}\right)
        +
        S\log_2\left(\frac{S}{\abs{\mathcal{K}}}\right)
    \right] \\
    &=
    \frac{\gamma^7}{9}
    \left[
        S\log_2\left(
            \frac{n}{4\abs{\mathcal{K}}}
        \right)
        +
        S\log_2\left(2 S\right)
        +
        S\log_2\left(\gamma^2\right)
    \right] \, .
\intertext{Now, the function $S\mapsto S\log_2\left(2 S\right)$ is never smaller than $-0.5=\log_2(1/\sqrt{2})$ for $S>0$. Hence:}
    &\geq
    \frac{\gamma^7}{9}
    \left[
        S\log_2\left(
            \frac{n}{4\abs{\mathcal{K}}}
        \right)
        +
        \log_2\left(\frac{1}{\sqrt{2}}\right)
        +
        S\log_2\left(\gamma^2\right)
    \right] \, .
\intertext{On the other hand, we know $\gamma^2<1$, so that $\log_2\left(\gamma^2\right)<0$. Since $S\leq \mu(\suppq)=1$:}
    &\geq
    \frac{\gamma^7}{9}
    \left[
        S\log_2\left(
            \frac{n}{4\abs{\mathcal{K}}}
        \right)
        +
        \log_2\left(\frac{\gamma^2}{\sqrt{2}}\right)
    \right] \, .
\intertext{Finally, consider the fact that $\abs{\mathcal{K}}\leq K$. Moreover, from Section \ref{sse:constant_same_distribution}, we know that when $d\mu=dF$ it holds that $R_{K,P_0}=\frac{1}{4K}$ (see equation \eqref{eq:const_abs_error}). Replacing above:}
    \sum_{k=1}^K E_k
    &\geq
    \frac{\gamma^7}{9}
    \left[
        S\log_2\left(n R_{K,P_0}\right)
        +
        \log_2\left(\frac{\gamma^2}{\sqrt{2}}\right)
    \right] \, .
\intertext{Now, by assumption, we have $n R_{K,P_0}\geq 1$. Hence $\log_2\left(n R_{K,P_0}\right)\geq 0$, and by \eqref{eq:prop-log-lower-bound-5} we conclude:}
    \sum_{k=1}^K E_k
    &\geq
    \frac{\gamma^7}{9}
    \left[
        \frac{1}{6}\log_2\left(n R_{K,P_0}\right)
        +
        \log_2\left(\frac{\gamma^2}{\sqrt{2}}\right)
    \right] \\
    &\geq
    \frac{\gamma^7}{54}
    \left[
        \log_2\left(n R_{K,P_0}\right)
        -
        \log_2\left(\frac{8}{\gamma^{12}}\right)
    \right] \, .
\end{align*}
The proof concludes by setting $C_1=\frac{\gamma^7}{54}$ and $C_2=\log_2\left(\frac{8}{\gamma^{12}}\right)$.
\end{proof}

\section{Auxiliary Results from Appendix \ref{appendix:prop-proof}}
\label{appendix:prop-aux}
\begin{lemma}[Paley–Zygmund inequality]
\label{lem:paley-zygmund-general}
Consider any measure $\mu$ over $\mathbb{R}$. Let $J \subseteq \mathbb{R}$ be a measurable set with $\mu(J) > 0$, and $G: J \to [0,\infty)$ be a measurable, non-negative function. Define its first and second moments over $J$ as
\begin{equation*}
    m_J = \int_J G(q) \, d\mu \, ,
    \quad\quad
    s_J = \int_J G(q)^2 d\mu \, .
\end{equation*}
Then, for any $\theta \in (0,1)$, the subset
\begin{equation*}
    J_\theta'
    =
    \left\{
        q \in J \>\text{ such that }\>
        G(q) \geq \theta \frac{m_J}{\mu(J)}
    \right\}
\end{equation*}
satisfies the lower bound
\begin{equation*}
    \mu(J_\theta')
    \geq
    (1 - \theta)^2 \, \frac{m_J^2}{s_J} \, .
\end{equation*}
\end{lemma}

\begin{proof}
Notice that the first moment can be written as
\begin{equation}
\label{eq:paley-zygmund-split}
    m_J
    =
    \int_J G(q) \, d\mu
    =
    \int_J G(q) \, \indf_{J_\theta'}(q) \, d\mu
    +
    \int_J G(q) \, \indf_{J\setminus\, J_\theta'}(q) \, d\mu \, ,
\end{equation}
where $\indf_U$ represents the indicator function of a set $U$. The first term on the right hand side can be bounded via Cauchy–Schwarz inequality:
\begin{align*}
    \int_J G(q) \, \indf_{J_\theta'}(q) \, d\mu
    &\leq
    \left(
        \int_J \indf_{J_\theta'}(q) \, d\mu
    \right)^{\frac{1}{2}}
    \left(
        \int_J G^2(q) d\mu
    \right)^{\frac{1}{2}}
    =
    \sqrt{\mu(J_\theta') \cdot s_J} \, .
\intertext{
    On the other hand, the last term in \eqref{eq:paley-zygmund-split} can be upper bounded as
}
    \int_J G(q) \, \indf_{J\setminus\, J_\theta'}(q) \, d\mu
    &=
    \int_{J\setminus\, J_\theta'} G(q) \, d\mu
    \leq
    \int_{J\setminus\, J_\theta'} \theta \frac{m_J}{\mu(J)} \, d\mu
    \leq
    \theta \frac{m_J}{\mu(J)} \mu(J)
    =
    \theta \cdot m_J \, .
\end{align*}
Replacing these bounds in equation \eqref{eq:paley-zygmund-split}, we get
\begin{equation*}
    m_J
    \leq
    \sqrt{\mu(J_\theta') \cdot s_J}
    +
    \theta \cdot m_J \, .
\end{equation*}
The result follows after subtracting $\theta \cdot m_J$ from both sides, squaring, and dividing by $s_J$.
\end{proof}

\begin{lemma}
\label{lem:bounds-first-second-moment}
Let $I \subseteq \mathbb{R}$ be a finite interval and $J \subseteq I$ a measurable subset of $I$ with positive measure $\mu(J) > 0$, where $\mu$ is defined by $d\mu = f\,dq$ for a density $f$ satisfying $0 < c_F \leq f(q) \leq C_F < \infty$, for all $q\in I$. Suppose $F: J \rightarrow \mathbb{R}$ is a measurable function such that there exists a point $q_0 \in I$ satisfying the bi-Lipschitz condition
\begin{equation}
\label{eq:lemma-bi-lipschitz-density}
    c_F\abs{q - q_0}
    \leq \abs{F(q) - F(q_0)}
    \leq C_F\abs{q - q_0}
\end{equation}
for all $q \in I$, where $c_F, C_F > 0$ are the same constants as above. Then, defining the first and second moments over $J$ with respect to $\mu$:
\begin{equation*}
    m_J = \int_J \abs{F(q) - F(q_0)} \, d\mu \,,
    \quad\quad
    s_J = \int_J \abs{F(q) - F(q_0)}^2 \, d\mu \,,
\end{equation*}
we have the following bounds:
\begin{equation*}
    m_J \geq \frac{c_F^2}{4C_F^2} \mu(J)^2 \,, \quad\quad
    s_J \leq \frac{C_F^3\,\mu(I)^3}{3 c_F^3} \,.
\end{equation*}
\end{lemma}

\begin{proof}
From the bi-Lipschitz condition \eqref{eq:lemma-bi-lipschitz-density}, we obtain the following bound on the first moment:
\begin{align*}
    m_J
    &=
    \int_J \abs{F(q) - F(q_0)}\,d\mu
    \geq
    c_F \int_J \abs{q - q_0} \, d\mu \, .
\intertext{Since $d\mu = f\,dq$ and $f \geq c_F$, we have:}
    m_J
    &\geq
    c_F \int_J \abs{q - q_0} \, f(q) dq
    \geq
    c_F^2 \int_J \abs{q - q_0} \, dq \,.
\intertext{
    The right-hand side is minimized when $J$ is a symmetric interval of length $\lambda(J)$ centered at $q_0$, where $\lambda$ denotes the Lebesgue measure. In this case, the integral becomes
}
    m_J
    &\geq
    c_F^2\int_{-\lambda(J)/2}^{\lambda(J)/2} \abs{q} \, dq
    =
    2\,c_F^2\int_0^{\lambda(J)/2} q \, dq
    =
    \frac{c_F^2}{4}\lambda(J)^2 \, .
\intertext{Finally, notice that $\lambda(J)=\int_J 1 dq \geq \frac{1}{C_F}\int_J f(q) dq=\frac{\mu(J)}{C_F}$, which gives the desired bound:}
    m_J
    &\geq
    \frac{c_F^2}{4 C_F^2}\mu(J)^2 \, .
\end{align*}

For the second moment $s_J$, we first use the upper bound in \eqref{eq:lemma-bi-lipschitz-density}, which implies
\begin{align*}
    s_J
    &=
    \int_J \abs{F(q) - F(q_0)}^2 d\mu
    \leq
    C_F^2 \int_J \abs{q - q_0}^2 \, d\mu \, .
\intertext{Since $d\mu = f\,dq$ and $f \leq C_F$, we have:}
    s_J
    &\leq
    C_F^2 \int_J \abs{q - q_0}^2 \, f(q) dq
    \leq
    C_F^3 \int_J \abs{q - q_0}^2 \, dq \,.
\intertext{
    This last quantity is maximized when $J = I$ and $q_0$ lies at one of its endpoints. This corresponds to the following integral:
}
    s_J
    &\leq
    C_F^3 \int_0^{\lambda(I)} q^2 dq
    =
    \frac{C_F^3}{3}\lambda(I)^3 \, ,
\intertext{where $\lambda$ denotes the Lebesgue measure. Finally, notice that $\lambda(I)=\int_I 1 dq \leq \frac{1}{c_F}\int_I f(q) dq=\frac{\mu(I)}{c_F}$, which gives the desired bound:}
    s_J
    &\leq
    \frac{C_F^3}{3 c_F^3}\mu(I)^3 \, .
\end{align*}
This completes the proof.
\end{proof}

\begin{lemma}[Entropy bound under partial mass]
\label{lem:partial-entropy-bound}
Let $K$ be a positive integer and $p_1, \dots, p_K$ be non-negative real numbers such that $\sum_{k=1}^K p_k = S > 0$. Then, it holds that
\begin{equation*}
    \sum_{k=1}^K p_k \log_2 (p_k)
    \geq
    S \log_2 \left( \frac{S}{K} \right) \, .
\end{equation*}
\end{lemma}

\begin{proof}
Let us define the normalized values $q_k = \frac{p_k}{S}$, for all $k = 1, \dots, K$. Then, it holds that $q_k \geq 0$ and $\sum_{k=1}^K q_k = 1$, i.e., $\{q_k\}_{k=1}^K$ constitutes a discrete probability distribution. We compute:
\begin{align}
    \sum_{k=1}^K p_k \log_2 (p_k)
    &=
    \sum_{k=1}^K S \, q_k \log_2(S \, q_k) \nonumber \\
    &=
    S \sum_{k=1}^K q_k \log_2 (S)
    + S \sum_{k=1}^K q_k \log_2 (q_k) \nonumber \\
    &=
    S \log_2 (S) + S \sum_{k=1}^K q_k \log_2 (q_k) \, . \label{eq:entropy-decomp}
\end{align}
The term $\sum_{k=1}^K q_k \log_2 (q_k)$ can be identified as the (negative) Shannon entropy, which is maximized for the uniform distribution (i.e., when $q_k=1/K$). This means that:
\begin{equation*}
    -\sum_{k=1}^K q_k \log_2 (q_k)
    \leq \log_2(K)
    \implies
    \sum_{k=1}^K q_k \log_2 (q_k)
    \geq \log_2\left(\frac{1}{K}\right) \, .
\end{equation*}
Applying this last inequality to equation \eqref{eq:entropy-decomp}, we get:
\begin{align*}
    \sum_{k=1}^K p_k \log_2 (p_k)
    &\geq
    S \log_2 (S) + S \log_2\left(\frac{1}{K}\right) \\
    &=
    S \log_2 \left( \frac{S}{K} \right) \, ,
\end{align*}
as claimed.
\end{proof}

\section{Proofs from Section \ref{sec:4-controlling-delta}}
\label{appendix:B}

An important result in the study of empirical processes is the Dvoretzky–Kiefer–Wolfowitz (DKW) inequality \cite{massart1990tight}, which we recap here. If $X_1,\ldots,X_n$ are i.i.d. and $\delta_n=F-F_n$ is the difference between the true CDF $F$ and the empirical CDF $F_n$, then for any $t>0$ it holds
\begin{equation*}
    \prob{\norm{\delta_n}_\infty > t} \leq 2e^{-2n t^2} \, .
\end{equation*}
We use this result in the proof of Lemma \ref{lem:expected-norms}.

\begin{proof}[Proof of Lemma \ref{lem:expected-norms}]
    Part \eqref{lem:expected-norms-a} follows from the tail formula for expected value and the DKW inequality \cite{massart1990tight}:
    \begin{align*}
        \expecs{X}{\norm{\delta_n}_\infty}
        &=
        \int_0^\infty \pr\left( \|\delta_n\|_\infty > t \right) dt
        \leq
        \int_0^\infty 2e^{-2nt^2} dt\,.
    \intertext{Substituting $u=2t\sqrt{n},\, du=2\sqrt{n}\,dt$ we get:}
        \expecs{X}{\norm{\delta_n}_\infty}
        &\leq
        \frac{1}{\sqrt{n}} \int_0^\infty e^{-\frac{u^2}{2}} du
        =
        \sqrt{\frac{2\pi}{n}} \underbrace{
            \int_0^\infty \frac{1}{\sqrt{2\pi}} e^{-\frac{u^2}{2}} du
        }_{=\, \frac{1}{2}}
        = \sqrt{\frac{\pi}{2n}}\,,
    \end{align*}
    where the last integral equals $1/2$ because the integrand is a standard Gaussian density. For Part \eqref{lem:expected-norms-b}, we use the fact that $\mu$ is a probability measure (i.e., $\mu(\suppq)=1$):
    \begin{equation*}
        \norm{\delta_n}_\mu
        =
        \int_\suppq \abs{\delta_n(q)} d\mu
        \leq
        \sup_{q\in\suppq} \, \abs{\delta_n(q)} \int_\suppq 1 \, d\mu
        \leq
        \norm{\delta_n}_\infty \cdot \mu(\suppq)
        =
        \norm{\delta_n}_\infty\,,
    \end{equation*}
    so that $\norm{\delta_n}_\mu \leq \norm{\delta_n}_\infty$ for any realization of $X_1, \ldots, X_n$\,. Applying monotonicity of expectation:
    \begin{equation*}
        \expecs{X}{\norm{\delta_n}_\mu}
        \leq
        \expecs{X}{\norm{\delta_n}_\infty}
        \leq
        \sqrt{\dfrac{\pi}{2n}}\,.
    \end{equation*}
\end{proof}

\begin{proof}[Proof of Lemma \ref{lem:existential-bound}]
    For Part \eqref{lem:existential-bound-a} we use Lemma~\ref{lem:expected-norms}, which shows that $\expecs{X}{\norm{\delta_n}_\mu} \leq \sqrt{\frac{\pi}{2n}}$\,. The expectation of a random variable cannot be strictly smaller than every individual value. Therefore, there must exist at least one realization of $\{X_i\}_{i=1}^n$ such that $\norm{\delta_n}_\mu \leq \sqrt{\frac{\pi}{2n}}$\,. For Part \eqref{lem:existential-bound-b}, we make use of a result by Csörgő and Faraway \cite{csorgHo1996exact}. Define the Cramér–von Mises statistic as
    \begin{equation*}
        \omega_n^2
        =
        n \int (F(q) - F_n(q))^2dF \, .
    \end{equation*}
    Then, for any $t\in \left[\frac{1}{12n}, \frac{n+3}{12n^2}\right]$ it holds (\cite[Equation 2.4]{csorgHo1996exact}):
    \begin{equation}
    \label{eq:cramer-von-mises-exact}
        \prob{\omega_n^2 \leq t}
        =
        \frac{n!\pi^{n/2}}{\Gamma(n/2+1)}
        \left(t-\frac{1}{12n}\right)^{n/2}
        > 0\,.
    \end{equation}
    Now, choose $t=\frac{1}{12n} + \frac{n}{12n^2}=\frac{1}{6n}$\,. Notice that:
    \begin{equation*}
        \omega_n^2
        =
        n \int (F(q) - F_n(q))^2dF
        =
        n \norm{\delta_n^2}_\mu
        \geq
        n \norm{\delta_n}_\mu^2\,,
    \end{equation*}
    where the last step is due to Jensen's inequality. Hence, we can write:
    \begin{equation*}
        \prob{\norm{\delta_n}_\mu \leq \frac{1}{\sqrt{6} \, n}}
        =
        \prob{n \norm{\delta_n}_\mu^2 \leq \frac{1}{6n}}
        \geq
        \prob{\omega_n^2 \leq \frac{1}{6n}}
        =
        \frac{n!\pi^{n/2}}{\Gamma(n/2+1)}
        \left(\frac{1}{12n}\right)^{n/2}
        >
        0\,,
    \end{equation*}
    where the last two steps are an application of equation \eqref{eq:cramer-von-mises-exact}. Since $\big\{ \norm{\delta_n}_\mu \leq \frac{1}{\sqrt{6} \, n} \big\}$ has probability greater than $0$, there must be a realization of $\{X_i\}_{i=1}^n$ where it occurs. This concludes the proof.
\end{proof}

\section{Lower Bounds on Query Time: Formal Results and Proofs}
\label{appendix:C}

We now derive concrete lower bounds on the query time $T(q)$ under the three search strategies described in Section~\ref{sec:4-query-time-reduction}. We do this by combining the approximation-theoretic results from Propositions~\ref{prop:basic-error-lower-bound} and~\ref{prop:log-error-lower-bound} with the deviation bounds from Lemmas~\ref{lem:expected-norms} and \ref{lem:existential-bound}. We begin with the case of linear search, for which $T(q) \geq \eps(q)$. From this inequality, we obtain:

\begin{corollary}[Lower bound for linear search]
\label{thm:linear-search}
    Let $X_1, \ldots, X_n$ be i.i.d.~samples from a distribution with CDF $F$, and let $q \sim \mu$ independently from the $\{X_i\}$. Let $I(A_n)$ be a learned index from the class $\lidx_{K,\cL}$\,, then the expected query time under linear search satisfies:
    \begin{equation*}
        \expecs{X,q}{T(q)}
        \geq
        n \Big[
            R_{K,\cL}(F) - {\textstyle \sqrt{\frac{\pi}{2n}}}
        \Big]\,.
    \end{equation*}
    Moreover, if $d\mu=dF$, then there exists a realization $X_1,\ldots,X_n$ such that
    \begin{equation*}
        \hspace{5pt}\expecs{q}{T(q)}
        \geq
        n \Big[
            R_{K,\cL}(F) - {\textstyle \frac{1}{\sqrt{6} \, n}}
        \Big]\,.
    \end{equation*}
\end{corollary}

\begin{proof}
    From Proposition \ref{prop:search-error-translation} we know $T(q)\geq \eps(q)$. Since $q$ is independent of the $\{X_i\}$:
    \begin{align*}
        \expecs{X,q}{T(q)}
        &\geq
        \expecs{X,q}{\eps(q)}
        =
        \expecs{X}{\expecs{q}{\eps(q)}} \, .
    \intertext{Now, we lower bound $\expecs{q}{\eps(q)}$ via Proposition~\ref{prop:basic-error-lower-bound}:}
        \expecs{X,q}{T(q)}
        &\geq
        \expecbs{X}{
            n \left[R_{K,\cL}(F) - \norm{\delta_n}_\mu\right]
        } \\
        &=
        n \Big[
            R_{K,\cL}(F) - \expecbs{X}{\norm{\delta_n}_\mu}
        \Big] \\
        &\geq
        n \Big[
            R_{K,\cL}(F) - {\textstyle \sqrt{\frac{\pi}{2n}}}
        \Big] \,,
    \end{align*}
    where the last step comes from Lemma~\ref{lem:expected-norms}\eqref{lem:expected-norms-b}. This proves the first part of the corollary. The second part follows directly from Proposition~\ref{prop:basic-error-lower-bound} and Lemma~\ref{lem:existential-bound}\eqref{lem:existential-bound-b}.
\end{proof}
    
Notice that the first part of the corollary provides a stronger type of guarantee, as it holds in full expectation over both the data and query distributions ($\expecs{X,q}{\cdot}$), and applies to any query distribution. In contrast, the second part assumes $d\mu = dF$ and yields only a worst-case guarantee over the data sample. However, this stronger assumption enables a sharper convergence rate: the deviation term scales as $1/n$ rather than $1/\sqrt{n}$\,. Similar trade-offs hold for other search algorithms.

In the case of exponential search, Proposition \ref{prop:search-error-translation} gives $T(q) \geq \tT(q) = \log_2\left(\max\{2, \eps(q)\}\right)$. For notational simplicity, we will write $T(q) \geq \log_2 \eps(q)$, where $\log_2(x)$ is assumed to output $1$ for $x\leq 2$. In this regime, we can derive the following lower bounds.

\begin{corollary}[Lower bound for exponential search]
\label{thm:exponential-search}
    Let $I(A_n)$ be a learned index from the class $\lidx_{K,\cL}$. Then, there exists a realization of $\{X_i\}_{i=1}^n$ and $q\in\suppq$ such that
    \begin{equation*}
        T(q)
        \geq
        \begin{cases}
            \log_2 \left(
                n \Big[R_{K,\cL}(F) - \sqrt{\frac{\pi}{2n}}\Big]
            \right)
            & \text{for arbitrary } \mu\,, \\ \\
            \log_2 \left(
                n \Big[R_{K,\cL}(F) - \frac{1}{\sqrt{6} \, n}\Big]
            \right)
            & \text{if } d\mu = dF\,.
        \end{cases}
    \end{equation*}
    Moreover, if the assumptions from Proposition \ref{prop:log-error-lower-bound} are satisfied, then there exists a realization of $\{X_i\}_{i=1}^n$ such that the expected query time under exponential search satisfies:
    \begin{equation*}
        \expecbs{q}{T(q)}
        \geq
        C_1 \left[
            \log_2 \left(n R_{K,\cL}(F)\right) - C_2
        \right] \, ,
    \end{equation*}
    where $C_1,C_2>0$ are constants independent of $n$ and $K$.
\end{corollary}

\begin{proof}
    We know $T(q)\geq \tT(q)$ (Proposition \ref{prop:search-error-translation}). Hence, the second part of the corollary is a direct application of Proposition \ref{prop:log-error-lower-bound}. For the first part, consider first an arbitrary $\mu$. Then by Lemma \ref{lem:existential-bound}\eqref{lem:existential-bound-a}, there is a realization of $\{X_i\}_{i=1}^n$ such that $\norm{\delta_n}_\mu \leq\sqrt{\frac{\pi}{2n}}$. At the same time, by Proposition \ref{prop:basic-error-lower-bound}:
    \begin{equation*}
        \expecbs{q}{\eps(q)} 
        \geq
        n \Big[R_{K,\cL}(F) - \norm{\delta_n}_\mu\Big]
        \geq
        n \Big[R_{K,\cL}(F) - {\textstyle\sqrt{\frac{\pi}{2n}}}\Big]
        \,,
    \end{equation*}
    so there exists at least one $q$ for which $\eps(q) \geq n \left[R_{K,\cL}(F) - \sqrt{\frac{\pi}{2n}}\right]$. Combining this with $T(q) \geq \log_2 \eps(q)$ proves the result. The case $d\mu=dF$ is proven in the same way, appealing to case \eqref{lem:existential-bound-b} of Lemma \ref{lem:existential-bound} to argue that there exists a realization of $\{X_i\}_{i=1}^n$ such that $\norm{\delta_n}_\mu \leq\frac{1}{\sqrt{6} \, n}$\,.
\end{proof}

Finally, for binary search, we consider the worst-case query time over both the data and the query value $q \in \suppq$. The cost is determined by the worst-case prediction error, as stated by Proposition \ref{prop:search-error-translation}.

\begin{corollary}[Lower bound for binary search]
    \label{thm:binary-search}
    Let $I(A_n)$ be a learned index from the class $\lidx_{K,\cL}$. Then, there exists a realization of $\{X_i\}_{i=1}^n$ and $q\in\suppq$ such that
    \begin{equation*}
        {\textstyle \sup_q T(q)}
        \geq
        \begin{cases}
            \log_2 \left(
                n \Big[R_{K,\cL}(F) - \sqrt{\frac{\pi}{2n}}\Big]
            \right)
            & \text{for arbitrary } \mu\,, \\ \\
            \log_2 \left(
                n \Big[R_{K,\cL}(F) - \frac{1}{\sqrt{6} \, n}\Big]
            \right)
            & \text{if } d\mu = dF\,.
        \end{cases}
    \end{equation*}
\end{corollary}

\begin{proof}
    Consider first the case of arbitrary $\mu$. By Lemma \ref{lem:existential-bound}\eqref{lem:existential-bound-a}, there exists a realization of $\{X_i\}_{i=1}^n$ such that $\norm{\delta_n}_\mu \leq\sqrt{\frac{\pi}{2n}}$\,. At the same time, by Proposition \ref{prop:basic-error-lower-bound}, we know
    \begin{equation*}
        \expecbs{q}{\eps(q)} 
        \geq
        n \Big[R_{K,\cL}(F) - \norm{\delta_n}_\mu\Big]
        \geq
        n \Big[R_{K,\cL}(F) - {\textstyle\sqrt{\frac{\pi}{2n}}}\Big]\,.
    \end{equation*}
    Hence, there exists at least one $q$ for which $\eps(q) \geq n \left[R_{K,\cL}(F) - \sqrt{\frac{\pi}{2n}}\right]$, which implies
    \begin{equation*}
        {\textstyle \sup_q \eps(q)}
        \geq
        n \Big[R_{K,\cL}(F) - {\textstyle\sqrt{\frac{\pi}{2n}}}\Big]\,.
    \end{equation*}
    Combining this with $\sup_q T(q) \geq \log_2 (\sup_q\eps(q))$ (Proposition \ref{prop:search-error-translation}) proves the result. The case $d\mu=dF$ is proven in the same way, appealing to case \eqref{lem:existential-bound-b} of Lemma \ref{lem:existential-bound} to argue that there exists a realization of $\{X_i\}_{i=1}^n$ such that $\norm{\delta_n}_\mu \leq\frac{1}{\sqrt{6} \, n}$\,.
\end{proof}

These results flesh out the general framework developed in Section~\ref{sec:4-new}, and express concrete lower bounds on query time in terms of $n$, the number of keys, and $R_{K,\cL}(F)$, the best approximation error achievable with a given model class.

\section{Piecewise Constant Approximation: Mismatched Dataset and Query Distributions}

\label{sse:constant_different_distribution}

In contrast to matching dataset and query distributions $f=g$ considered in Section~\ref{sse:constant_same_distribution}, the case $f \neq g$ is relatively more involved; here too, however, the problem of approximating $F$ can be reduced to optimal quantization. The only difference is that the quantized random variable is no longer uniformly distributed, and our results hold asymptotically for $K\to\infty$ as opposed to any given $K$ as in Section~\ref{sse:constant_same_distribution}.

The results in this section assume that $F$ is invertible. This is a relatively mild condition, which holds for most standard examples of continuous distributions. Moreover, this assumption can be relaxed: if $F$ is only invertible on a subinterval $[a, b] \subseteq \suppq$, the analysis still applies by conditioning the query distribution to $[a,b]$. This affects only the scaling constants, through changes in the conditional query distribution, but preserves the asymptotic behavior in $K$. For clarity of exposition, we focus on the case where $F$ is invertible over all of $\suppq$.

\begin{proposition}[Reduction to Quantization]
\label{prop:quantization_equivalence}
    Let $h\in \lidx_{K, P_0}$ and $Y = F(Q)$ with $Q \sim \mu$ and $d\mu = g \, dq$. Then
    \begin{equation*}
        \norm{F - h}_{\mu,r} \defeq \int_\suppq
    \lvert F(q) - h(q) \rvert^r g(q) dq = \mathbb{E}[|Y - Q_K(Y)|^r]  \,.
    \end{equation*}
\end{proposition}
\begin{proof}
    As in~(\ref{eq:norm_constant_uniform}), we expand $\norm{F - h}_{\mu,r}$ as
    \begin{equation*}
\label{eq:norm_constant_uniform_mismatched}
    \norm{F - h}_{\mu,r}
    =
    \sum_{k=1}^K
    \int_{I_k} \lvert F(q) - c_k \rvert^r g(q) dq\,.
\end{equation*}
Introducing $y = F(q),~dy = f(q)dq,~q=F^{-1}(y),$ and $dq = 1/f(q) \, dy$, we rewrite this as
\begin{equation*}
    \sum_{k=1}^K
    \int_{J_k} \lvert y - c_k \rvert^r \cdot \frac{g\left(F^{-1}(y)\right)}{f\left(F^{-1}(y)\right)} dy = \mathbb{E}[|Y - Q_K(Y)|^r] \,,
\end{equation*}
 where the last equality follows because we recognize that the second term in the integral is the probability density function (PDF) of $Y = F(Q)$ when $Q$ distributes according to the PDF $g$:
\begin{equation}
\label{eq:pdf_of_a_function}
    f_Y(y) = g\left(F^{-1}(y)\right) \cdot \left| \frac{d}{dy} F^{-1}(y)\right| = g\left(F^{-1}(y)\right) \cdot \left| \frac{1}{F'\left( F^{-1} (y) \right)} \right| = \frac{g\left(F^{-1}(y)\right)}{f\left(F^{-1}(y)\right)} \,;
\end{equation}
the first equality expands the density of a function of a random variable, and the second the derivative of an inverse function.
We remark that this requires the CDF $F$ to be invertible.
\end{proof}

Proposition~\ref{prop:quantization_equivalence} implies that approximating $F$ reduces to quantizing a random variable $Y=F(Q)$ distributed according to~(\ref{eq:pdf_of_a_function}) over $[0,1]$. When $f \neq g$, this distribution is no longer uniform, and uniform quantization is no longer optimal. To carry on with our visual analogy in Figure ~\ref{fig:cdf_figure}, one can imagine the density of the green dashed lines below $F$ no longer proportional to $f$, and therefore the ones above $F$ no longer uniform, in which case the approximating values $\{ c_k \}$ would have to be shifted accordingly.

We are now in a position to leverage an important result in quantization theory. The asymptotic quantization error is known to satisfy~\cite[Theorem 6.2]{graf2000foundations}
\begin{equation*}
    \lim_{K \to \infty} K^r ~ T_{K,r}(Y) = \frac{1}{(1+r)2^r} \left( \int f_Y(y)^{1 / (1+r)} dy \right)^{1+r}\,.
\end{equation*}
The right-hand side of the equality is known as the $r$-th quantization coefficient and is a fundamental property of a probability distribution, similar to its moments; quantization coefficients play an important role in quantization theory. Therefore, for the absolute error
\begin{equation*}
    \lim_{K \to \infty} K ~ R_{K,P_0}(F) = \lim_{K \to \infty} K ~ T_{K,1}(Y) = \frac{1}{4} \left( \int \sqrt{f_Y(y)} dy \right)^2\,.
\end{equation*}
In other words, asymptotically, the absolute error decreases linearly with $K$, just like in the case of the matching distribution considered in Section~\ref{sse:constant_same_distribution}, yielding $R_{K,\cL}(F) = \Omega\left(1/K\right)$ when $\cL = P_0$\,.


\section{Proofs from Section \ref{sec:7}}
\label{appendix:D}
\begin{proof}[Proof of Lemma \ref{lemma:width-inequality}]
\label{appendix:D-proof-fact}
    $R_{K, \cL}(F)$ can be written as
    \begin{align*}
        R_{K, \cL}(F)
        &=
        \inf_{h \in \lidx_{K,\cL}} \norm{F - h}_\mu
        =
        \inf_{\cP\in\mathbf{P}_K(\suppq)}
        \inf_{h \in \lidx_{\cP,\cL}} \norm{F - h}_\mu\,,
    \intertext{
        where $\inf_{\cP\in\mathbf{P}_K(\suppq)}$ is taken over the set $\mathbf{P}_K(\suppq)$ of all partitions of $\suppq$ into $K$ intervals. By definition of infimum, for any $t>0$ there exists $\cP_t \in \mathbf{P}_K(\suppq)$ such that
    }
        R_{K, \cL}(F)
        &\geq
        \inf_{h \in \lidx_{\cP_t,\cL}} \norm{F - h}_\mu - t \, .
    \intertext{Taking supremum with respect to $F$:}
        \sup_{F\in\cF} R_{K, \cL}(F)
        &\geq
        \sup_{F\in\cF} \Inf_{h \in \lidx_{\cP_t,\cL}} \norm{F - h}_\mu - t \, .
    \intertext{
        Now, $\lidx_{\cP_t,\cL}$ is a linear space of dimension at most $\kappa = K\cdot\dim(\cL)$. Hence, the last expression can be lower bounded by taking infimum over all such spaces:
    }
        \sup_{F\in\cF} R_{K, \cL}(F)
        &\geq
        \Inf_{H \in \cH_\kappa} \sup_{F\in\cF} \Inf_{h \in H}
        \norm{F - h}_\mu - t
        = d_\kappa(\cF) - t \,.
    \end{align*}
    Notice that the $\inf$-$\sup$-$\inf$ sequence now has no dependence on $t$. Therefore, since the inequality holds for all $t > 0$, it holds for $t=0$ too, by standard properties of the real numbers.
\end{proof}

\subsection{Upper Bounds for Piecewise Linear Approximation}
\label{appendix:D-upper-bounds}

We briefly sketch standard upper bounds for the approximation error $R_{K,\cL}(F)$ when $\cL = P_1$, the class of degree-one polynomials, under two common regularity assumptions on the target function $F$. Throughout, we assume the domain is $[0,1]$, the measure $\mu$ is the Lebesgue measure, and the interval is divided into $K$ uniform subintervals $[a_i, a_{i+1}]$, where $a_i = \frac{i-1}{K}$ and $a_{i+1} = \frac{i}{K}$\,. We denote the length of each subinterval by $\ell = 1/K$.

\subparagraph{Case 1: $F$ is twice continuously differentiable.}
Let us approximate $F$ over each interval $[a_i, a_{i+1}]$ using a linear function based on a first-order Taylor expansion:
\begin{equation*}
    L_i(x) = F(a_i) + (x - a_i) F'(a_i)\,.
\end{equation*}
Since $F$ is $C^2([0,1])$, we can bound the pointwise error using Taylor's theorem:
\begin{equation*}
    \abs{F(x) - L_i(x)} \leq \frac{1}{2} M \cdot (x - a_i)^2\,,
\end{equation*}
where $M=\max_{x \in [0, 1]} \abs{F''(x)}$. We know that $M$ exists and is bounded because $F''$ is continuous and $[0,1]$ is a compact interval. Integrating this over each interval yields:
\begin{equation*}
    \int_{a_i}^{a_{i+1}} \abs{F(x) - L_i(x)} dx
    \leq
    \frac{1}{2} M \int_{a_i}^{a_{i+1}} (x - a_i)^2 dx
    =
    \frac{1}{6} M \ell^3\,.
\end{equation*}
Finally, summing over all $K$ intervals gives:
\begin{equation*}
    \norm{F - h}_{\mu} \leq K \cdot \frac{1}{6} M \ell^3 = O(1/K^2)\,.
\end{equation*}

\subparagraph{Case 2: $F$ is $M$-Lipschitz.}
Let us now approximate $F$ using linear interpolation between endpoints:
\begin{equation*}
    L_i(x) = F(a_i) + \frac{F(a_{i+1}) - F(a_i)}{a_{i+1} - a_i} \cdot (x - a_i)\,.
\end{equation*}
By the Lipschitz property, for any $x \in [a_i, a_{i+1}]$ we have:
\begin{align*}
    \abs{F(x) - L_i(x)}
    &=
    \Big\lvert
        F(x) - F(a_i)
        - \frac{F(a_{i+1}) - F(a_i)}{a_{i+1} - a_i} (x - a_i)
    \Big\rvert \\
    &=
    \Big\lvert
        \frac{F(x) - F(a_i)}{x - a_i}
        - \frac{F(a_{i+1}) - F(a_i)}{a_{i+1} - a_i}
    \Big\rvert \cdot \abs{x - a_i} \\
    &\leq
    \left(
        \Big\lvert \frac{F(x) - F(a_i)}{x - a_i} \Big\rvert
        + \Big\lvert \frac{F(a_{i+1}) - F(a_i)}{a_{i+1} - a_i} \Big\rvert
    \right) \cdot \abs{x - a_i} \\
    &\leq (M + M) \abs{x - a_i}\,.
\end{align*}
Therefore, the absolute error over each interval is at most $2M\abs{x - a_i}$. Integrating over each interval:
\begin{equation*}
    \int_{a_i}^{a_{i+1}} \abs{F(x) - L_i(x)} dx
    \leq
    2M \int_{a_i}^{a_{i+1}} \abs{x - a_i} dx
    =
    M \ell^2\,.
\end{equation*}
Finally, summing over all $K$ intervals gives:
\begin{equation*}
    \norm{F - h}_{\mu} \leq K \cdot M\ell^2 = O(1/K)\,.
\end{equation*}

These asymptotic rates for $R_{K,P_1}(F)$---i.e., $O(1/K)$ under Lipschitz continuity and $O(1/K^2)$ under $C^2$ assumption---hold more generally for functions defined on any compact interval, and under any measure $\mu$ that has a bounded density. The specific constants may depend on the domain and the density of $\mu$, but the convergence rates remain unaffected. Thus, piecewise linear functions can provide different approximation guarantees under specific regularity assumptions.

\subsection{Construction of Adversarial CDF with Respect to \texorpdfstring{$K$}{K}}
\label{appendix:D-adversarial-cdf}

We will illustrate the case $\suppq = [0,1]$, with $\mu$ being the uniform distribution. This will allows us to perform some exact calculations, but the asymptotic properties we derive in this section, and the argument we use to prove them, apply for very general choices of $\suppq$ and $\mu$. Consider the function $F(x)=x^2$ and the optimal approximation error $R_{1,P_1}(F)$ by a single linear function. Notice that $x^2$ over $[0,1]$ is a valid CDF. From \cite[Corollary 3.4.1]{rivlin1981introduction}, it is easy to deduce that
\begin{equation}
\label{eq:optimal-error-x-squared}
    R_{1,P_1}(F) = 1/16\,,
\end{equation}
which is attained by $h(x)=x-3/16$.
Now, take a positive integer $M$. The intuition is that we can construct a staircase-type function $F_M$, with $M$ total steps and each step $F_M^i$ being equal to $F(x)=x^2$ after appropriate translation and scaling, such that the full $F_M$ constitutes a valid CDF and is hard to approximate. For any $i=0,\ldots,M-1$ define $F_M^i$ as:
\begin{align*}
    F_M^i(x)
    &=
    \frac{1}{M}\left[
        M\left(x - \frac{i}{M}\right)
    \right]^2 + \frac{i}{M}
    \quad\text{for}\quad
    x\in I_i=\left[\frac{i}{M},\frac{i+1}{M}\right].
\intertext{Now, the full $F_M$ is defined as}
    F_M(x)
    &=
    \sum_{i=0}^{M-1} F_M^i(x)\indf_{I_i}(x)
    \quad\text{for}\quad
    x\in [0,1]\,.
\end{align*}
It is easy to check that $F_M$ is continuous, increasing, and satisfies $F(0)=0$, $F(1)=1$, thus being a valid CDF on $[0,1]$. Consider now any linear approximation $L_i\in P_1$ for the function $F_M^i$ over the interval $I_i$. The approximation error is then given by:
\begin{align}
    E_i
    &=
    \int_{\frac{i}{M}}^{\frac{i+1}{M}}
    \abs{F_M^i(x) - L_i(x)}dx \nonumber \\
    &= 
    \int_{\frac{i}{M}}^{\frac{i+1}{M}}\left\lvert
        \frac{1}{M}\left[M\left(x - \frac{i}{M}\right)\right]^2
        + \frac{i}{M} - L_i(x)
    \right\rvert dx\,. \nonumber
\intertext{Substituting $u=M(x-i/M),~du=Mdx$:}
    &= 
    \frac{1}{M}\int_0^1\left\lvert
        \frac{u^2}{M} + \frac{i}{M} - L_i\left(\frac{u+i}{M}\right)
    \right\rvert du \nonumber \\
    &= 
    \frac{1}{M^2}\int_0^1\left\lvert
        u^2 + i - M L_i\left(\frac{u+i}{M}\right)
    \right\rvert du\,. \nonumber \\
\intertext{We focus now on the integral term. Since $i-M L_i((u+1)/M)\in P_1$, it cannot induce lower error when approximating $u^2$ than the optimal approximator $h(u)=u-3/16$, hence by equation \eqref{eq:optimal-error-x-squared}:}
    E_i
    &\geq 
    \frac{1}{16M^2}\,. \label{eq:error-one-step}
\end{align}
We now use this to derive a lower bound on $R_{K,P_1}(F_M)$. Since $F_M$ consists of $M$ disjoint quadratic steps, each over an interval of length $1/M$, a $K$-piecewise linear function can assign at most $K$ linear segments across all $M$ sub-parabolas. Therefore, at least $M - K + 1$ sub-parabolas must be covered using only one linear piece each, and on each corresponding interval, bound \eqref{eq:error-one-step} applies. Summing over these $M - K + 1$ intervals, we obtain
\begin{align*}
    R_{K,P_1}(F_M)
    &\geq
    \frac{M-K+1}{16M^2} \, .
\intertext{This expression is maximized when $M = 2(K-1)$, yielding}
    R_{K,P_1}(F_{2(K-1)})
    &\geq
    \frac{1}{64(K-1)}
    =
    \Omega\left( \frac{1}{K} \right).
\end{align*}

This construction shows that even when $F$ is a valid CDF, it is possible to exhibit $\Omega(1/K)$ lower bounds for approximation using $K$-piecewise linear functions. In particular, the worst-case behavior of $R_{K,P_1}(F)$ is no better than that of piecewise constant models, despite the greater expressive power of linear segments. This highlights the critical role of regularity assumptions on $F$: without sufficient smoothness, more expressive models do not necessarily yield improved approximation rates.

\subsection{Kolmogorov Widths for Lipschitz CDFs}
\label{appendix:D-lipschitz-cdf}

Without loss of generality, we consider $[a,b]=[0,1]$ and restrict our attention to the class $\mathcal{F}$ of Lipschitz continuous functions with Lipschitz constant $1$, that is:
\begin{equation*}
    \cF
    =
    \left\{ F \in L^1(\mu) : \abs{F(x)-F(y)} \leq \abs{x-y} \right\}.
\end{equation*}
Let $d_\kappa(\mathcal{F})$ denote the Kolmogorov width of $\mathcal{F}$ in $L^1(\mu)$, then:
\begin{equation}
\label{eq:kolmogorov-width-def}
    d_\kappa(\cF)
    =
    \Inf_{H\in\cH_\kappa}
    \sup_{F\in\cF} \Inf_{h \in H} \norm{F - h}_\mu
    = \Omega(1/\kappa)
    \, ,
\end{equation}
where $d_\kappa(\mathcal{F}) = \Omega(1/\kappa)$ comes from classic results referenced in Section \ref{sec:7-n-widths}. However, not all functions in $\mathcal{F}$ are valid cumulative distribution functions (CDFs), as they need not be nonnegative or monotonically increasing. We now show that from any $F \in \mathcal{F}$, we can construct a valid CDF $\widehat{F}$ that remains equally hard to approximate up to constants.

For notational simplicity, denote the distance from $F$ to the set $H$ as $\norm{F-H}_\mu$, that is, $\norm{F-H}_\mu = \inf_{h \in H} \norm{F - h}_\mu$. Notice that the outermost infimum in equation \eqref{eq:kolmogorov-width-def} is taken over all linear subspaces $H\subseteq L^1(\mu)$ of dimension at most $\kappa$. Hence, the bound applies to any such space. Consider then any $H^\ast$ which includes at least all constant and linear functions (this can always be done as long as $\kappa\geq 2$, since $\spn\{1,x\}$ is an example of such a space). Now, fix $t>0$. By equation \eqref{eq:kolmogorov-width-def} and the definition of supremum, there exists $F$ such that
\begin{equation*}
    \norm{F-H^\ast}_\mu
    \geq
    \frac{C}{\kappa} - t\,,
\end{equation*}
where $C>0$ is some positive constant independent of $t$.

\subparagraph{Step 1: Constructing a strictly increasing function.}
\label{appendix:D-step-1}
It is well-known that $\cF$ coincides with the Sobolev ball $W^{1,\infty}([0,1])$ of weakly-differentiable functions, such that their weak derivatives are essentially bounded by $1$ \cite[Chapter 5.8, Theorem 4]{evans10}. This implies that $F'\geq -1$ up to a set of measure $0$. Now, define
\begin{align*}
    \widetilde{F}(x)
    &= F(x) + 2x\,.
\intertext{
    Then, $\widetilde{F}$ is monotonic on $[0,1]$, because its weak derivative $\widetilde{F}'$ is positive:
}
    \widetilde{F}'(x)
    &= F'(x) + 2 \geq -1 + 2 = 1\,.
\end{align*}
Therefore, $\widetilde{F}$ is strictly increasing, and it is continuous since both $F$ and $2x$ are Lipschitz continuous.

\subparagraph{Step 2: Rescaling to a valid CDF.}
\label{appendix:D-step-2}
We now normalize $\widetilde{F}$ to obtain a function $\widehat{F}$ satisfying both $\widehat{F}(0)=0$ and $\widehat{F}(1)=1$:
\begin{equation*}
    \widehat{F}(x)
    =
    \frac{ \widetilde{F}(x) - \widetilde{F}(0) }
         { \widetilde{F}(1) - \widetilde{F}(0) }
    = \frac{ F(x) + 2x - F(0) }{ F(1) + 2 - F(0) }\,.
\end{equation*}
$\widehat{F}$ is strictly increasing and continuous in $[0,1]$, because $ \widetilde{F}$ is. Hence, it is a valid CDF. Moreover, we can upper and lower bound the scaling constant:
\begin{equation}
\label{eq:normalizing-constant-bounds}
    (\widetilde{F}(1) - \widetilde{F}(0))
    =
    (F(1) + 2 - F(0))
    \in
    [1, 3]\,,
\end{equation}
since  $(F(1) - F(0)) \in [-1, 1]$ by the $1$-Lipschitz property.

\subparagraph{Step 3: Preserving approximation difficulty.} Recall that
\begin{align}
    \norm{F-H^\ast}_\mu
    &\geq
    \frac{C}{\kappa} - t\,. \nonumber
\intertext{Now, denote $M=F(1) + 2 - F(0)$ and consider}
    \norm{ \widehat{F}-H^\ast }_\mu
    &=
    \left\lVert
        \frac{ F + 2x - F(0) }{M} - H^\ast
    \right\rVert_\mu
    =
    \frac{1}{M}
    \norm{F + 2x - F(0) - M\cdot H^\ast}_\mu \, . \nonumber
\intertext{Since $H^\ast$ is a linear subspace, it is invariant under scalar multiplication, so that $M\cdot H^\ast = H^\ast$ provided $M\neq 0$. Furthermore, it also holds that $2x - F(0) - M\cdot H^\ast = H^\ast$, because we have assumed that $H^\ast$ spans the set of constant and linear functions. By equation \eqref{eq:normalizing-constant-bounds}, we know that $M\leq 3$, so all in all we can conclude that:}
    \norm{ \widehat{F}-H^\ast }_\mu
    &=
    \frac{1}{M}
    \norm{F - H^\ast}_\mu
    \geq
    \frac{1}{3}
    \norm{F - H^\ast}_\mu
    \geq
    \frac{C}{3\kappa} - \frac{t}{3}\,. \label{eq:kolmogorov-width-for-cdf}
\intertext{
    Let $\widehat{\cF} = \{ \widehat{F} \text{ constructed from } F\in \mathcal{F} \text{ as described in Steps \hyperref[appendix:D-step-1]{1} and \hyperref[appendix:D-step-2]{2}} \}$. Then, taking supremum over that class in equation \eqref{eq:kolmogorov-width-for-cdf} we get:
}
    \sup_{\widehat{F}\in \widehat{\cF}} \norm{ \widehat{F}-H^\ast }_\mu
    &\geq
    \frac{C}{3\kappa} - \frac{t}{3}\,, \nonumber
\intertext{
    where $C$ does not depend on $t$. Since this holds for any $t > 0$, it holds for $t=0$ as well. Now, taking infimum with respect to all linear spaces of dimension $\kappa$, we can conclude:
}
    d_\kappa(\widehat{\cF})
    &\geq
    \frac{C}{3\kappa}
    =
    \Omega\left(\frac{1}{\kappa}\right). \nonumber
\end{align}
This construction shows that any lower bound on approximation for Lipschitz continuous functions extends to a lower bound for valid CDFs, up to a constant multiplicative factor. Therefore, all conclusions from Section \ref{sec:7-n-widths} remain valid even when restricted specifically to CDFs.

\end{document}